\documentclass[12pt]{article}
\usepackage{color}
\usepackage{amsmath, amsthm, amssymb, bm, setspace,bigints,algorithmicx,algorithm}
\usepackage[margin=1 in]{geometry}
\usepackage{times}
\usepackage{hyperref}
\usepackage{subcaption}
\usepackage{bm}
\usepackage{bbm}
\usepackage{natbib}
\usepackage[pdftex]{graphicx}
\usepackage{pdfpages}
\usepackage{setspace}
\usepackage[affil-it]{authblk}
\newtheorem{theorem}{Theorem}
\newtheorem{lemma}{Lemma}
\newtheorem{proposition}[theorem]{Proposition}
\newtheorem{corollary}[theorem]{Corollary}

\def\T{{ \mathrm{\scriptscriptstyle T} }}

\def\mr{\mathrm}

\newcommand{\bige} {\mbox{\boldmath ${\bm \epsilon}$}}

\newcommand{\norm} {\|}

\newcommand{\bk}{{\bf k}}
\newcommand{\bt}{{\bf t}}

\newcommand{\be} {\begin{eqnarray*}}
\newcommand{\ee} {\end{eqnarray*}}

\begin{document}


\title{Scalable Bayesian Variable Selection Using Nonlocal Prior Densities in Ultrahigh-dimensional Settings}

\author{Minsuk Shin%
\thanks{Electronic address: \texttt{minsuk@stat.tamu.edu}}}
\author{Anirban Bhattacharya%
\thanks{Electronic address: \texttt{anirbanb@stat.tamu.edu}}}
\author{Valen E. Johnson%
\thanks{Electronic address: \texttt{vjohnson@stat.tamu.edu}}}
  
\affil{Department of Statistics, Texas A\&M University, Texas, U.S.A}
\date{}

\maketitle

\begin{abstract}
 Bayesian model selection procedures based on nonlocal alternative prior densities are extended to ultrahigh dimensional settings and compared to other variable selection procedures using precision-recall curves.  Variable selection procedures included in these comparisons include methods based on $g$-priors, reciprocal lasso, adaptive lasso, scad, and minimax concave penalty criteria.  The use of precision-recall curves eliminates the sensitivity of our conclusions to the choice of tuning parameters.  We find that Bayesian variable selection procedures based on nonlocal priors are competitive to all other procedures in a range of simulation scenarios, and we subsequently explain this favorable performance through a theoretical examination of their consistency properties.  When certain regularity conditions apply, we demonstrate that the nonlocal procedures are consistent for linear models even when the number of covariates $p$ increases sub-exponentially with the sample size $n$.  A model selection procedure based on Zellner's $g$-prior is also found to be competitive with penalized likelihood methods in identifying the true model, but the posterior distribution on the model space induced by this method is much more dispersed than the posterior distribution induced on the model space by the nonlocal prior methods. We investigate the asymptotic form of the marginal likelihood based on the nonlocal priors and show that it attains a unique term that cannot be derived from the other Bayesian model selection procedures. We also propose a scalable and efficient algorithm called Simplified Shotgun Stochastic Search with Screening (S5) to explore the enormous model space, and we show that S5 dramatically reduces the computing time without losing the capacity to search the interesting region in the model space, at least in the simulation settings considered.  The S5 algorithm is available in an \verb R ~package {\it BayesS5} on \texttt{CRAN}.\end{abstract}


\section{Introduction}
In the context of hypothesis testing, \citet{johnson2010use} defined nonlocal (alternative) priors as densities that are exactly zero whenever a model parameter equals its null value. Nonlocal priors were extended to model selection problems in \citet{Johnson2012}, where product moment (pMoM) prior and product inverse moment (piMoM) prior densities were introduced as priors on a vector of regression coefficients.  In $p\leq n$ settings, model selection procedures based on these priors were demonstrated to have a strong model selection property: the posterior probability of the true model converges to 1 as the sample size $n$ increases.  
More recently, \citet{rossell2013high} and \citet{rossell2015non} proposed product exponential moment (peMoM) prior densities that have similar behavior to piMoM densities near the origin. However, the behavior of nonlocal priors in $p \gg n$ settings remains understudied to date  (particularly in comparison to other commonly used variables selection procedures), which serves as the motivation for this article.

We undertook a detailed simulation study to compare the performance of nonlocal priors in $p \gg n$ settings under sparsity with a 
host of penalization methods including the least absolute shrinkage and selection operator (lasso; \citet{tibshirani1996regression}), smoothly clipped absolute deviation (scad; \cite{fan2001variable}), adaptive lasso \citep{zou2006adaptive}, minimum convex penalty (mcp; \cite{zhang2010nearly}), and the reciprocal lasso (rlasso), recently been proposed by \citet{song2015high}. The penalty function of the rlasso is equivalent to the negative log-kernel of nonlocal prior densities; further connections are described in Section \ref{sec:rlasso}. As a natural Bayesian competitor, we also considered the widely used $g$-prior \citep{zellner1986,liang2008mixtures}, which is a local prior in the sense of \citet{johnson2010use}. We used precision-recall curves \citep{davis2006relationship} as a basis for comparison  between methods. These curves eliminate the effect of the choice of tuning parameters for each method so that the comparison across different methods can be transparent. It has been argued \citep{davis2006relationship} that in cases where only a tiny proportion of variables are significant, precision-recall curves are more appropriate { tools for comparison than are the more widely used receiver operating characteristic curves. While the ROC curves present a trade-off between the type I error and the power of a decision procedure, precision-recall curves examine the trade-off between the power and the false discovery rate.

{ Our studies indicate that Bayesian procedures based on nonlocal priors and the $g$-prior perform better than penalized likelihood approaches in a sense that they achieve a lower false discovery rate, while maintaining the same power of the decision procedure. Posterior distributions on the model space based on nonlocal priors were found to be more tightly concentrated around the  maximum a posteriori model than the posterior based on $g$-priors, implying that they had a faster rate of posterior concentration. We also identified the oracle hyperparameter that maximizes the posterior probability of the true model for the Bayesian procedures. The growth-rate of these oracle hyperparameters with $p$ also offers an interesting contrast between nonlocal and local priors. In the case of $g$-priors, the oracle value of $g$ varied between $7.83\times 10^8$ and $4.29\times 10^{13}$ as $p$ ranged between $1000$ and $20000$. For the same range of $p$, the oracle value of $\tau$ varied between $1.97$ and $3.60$, where $\tau$ is the tuning parameter for nonlocal priors described in Section 2.} \cite{george2000calibration} argued from a minimax perspective that the $g$ parameter should satisfy $g \asymp p^2$, which explains the large values of the optimal $g$. 
However, using asymptotic arguments to obtain default hyperparameters is difficult because the constant of proportionality is typically unknown. Moreover, when $g$ is very large, the $g$-prior assigns negligible prior mass at the origin, essentially resulting in a nonlocal like prior. A similar point can be made about the recently proposed Bayesian shrinking and diffusing (BASAD) priors \citep{Narisetty2014}. On the other hand, the optimal hyperparameter value for the nonlocal priors is stable with increasing $p$, growing at a very slow rate. 

Motivated by this empirical finding, we studied properties of two classes of nonlocal priors allowing the hyperparameter $\tau$ to scale with $p$. Using a fixed value of $\tau$, it seems that strong model selection consistency is possible only when $p \leq n$ \citep{Johnson2012}. In this article, we establish that nonlocal priors can achieve strong model selection consistency even when the number of variables $p$ increases sub-exponentially in the sample size $n$, provided that the hyperparameter $\tau$ is asymptotically larger than $\log p$.  This theoretical result is consistent with our empirical finding.

\section{Nonlocal prior densities for regression coefficients}
We consider the standard setup of a Gaussian linear regression model with a univariate response and $p$ candidate predictors. Let $y = (y_1, \ldots, y_n)^{\T}$ denote a vector of responses for $n$ individuals and ${ X}$ an $n\times p$ matrix of covariates. We denote a model by $\bk= 
\{ k_1,\ldots, k_{|\bk|} \}$, with $1 \leq k_1 < \ldots < k_{|\bk|} \leq p$. Given a model $\bk$, let ${ X_\bk}$ denote the design matrix formed from the columns of ${X}_n$ corresponding to model $\bk$ and $\beta_\bk=(\beta_{k,1}, \ldots ,\beta_{k,{|\bk|}})^{\T}$ the regression coefficient for model $\bk$.
Under each model \bk, the linear regression model for the data is 
\begin{eqnarray}\label{eq:like}
y  = {X}_{\bk} {\beta}_\bk+\epsilon,
\end{eqnarray}
where $\epsilon \sim N_n(0,\sigma^2\mr{I}_n)$.
Let $\bt$ denote the true, or data-generating model and let ${\beta}_\bt^0$ be the true regression coefficient under model $\bt$. We assume that the true model is fixed but unknown. 

Given a model $\bk$, the product exponential moment (peMoM) prior density \citep{rossell2013high,rossell2015non} for the vector of regression coefficients $\beta_{\bk}$ is defined as 
\begin{align}\label{eq:general_nonlocal}
\pi(\beta_{\bk} \mid \sigma^2, \tau, \bk) = C^{-|\bk|} \prod_{j=1}^{|\bk|}\exp\{-\beta_{\bk,j}^2 /(2\sigma^2\tau) -\tau/\beta_{\bk,j}^2 \}.
\end{align}
The normalizing constant $C$ can be explicitly calculated as 
\begin{align}\label{eq:normalizer}
C = \int_{ -\infty}^{\infty} \exp\{-t^2/(2 \sigma^2 \tau)-\tau/t^2\} dt= (2\pi\sigma^2\tau)^{1/2} \exp \{ - (2/\sigma^2)^{1/2} \}, 
\end{align}
since $\int \exp\{-\mu/t^2-\zeta t^2\}dt=(\pi/\zeta)^{1/2}\exp\{ -2(\mu\zeta)^{1/2} \}$.

Second, for a fixed positive integer $r$, the product inverse-moment (piMoM) prior density \citep{Johnson2012} for $\beta_{\bk}$ is given by 
\begin{align}\label{eq:mod_piMoM}
\pi(\beta_{\bk} \mid \sigma^2, \tau, \bk)=C^{*-|\bk|}\prod_{j=1}^{|\bk|}[(\beta_{\bk,j})^{-2r}\exp\{ -\tau/\beta_{\bk,j}^2 \} ],
\end{align}
where $C^{*}=\tau^{-r+1/2}\Gamma(r-1/2)$ for $r > 1/2$, where $\Gamma(\cdot)$ is the gamma function.

\begin{figure}
\centering
\includegraphics[height=6cm, width=9cm]{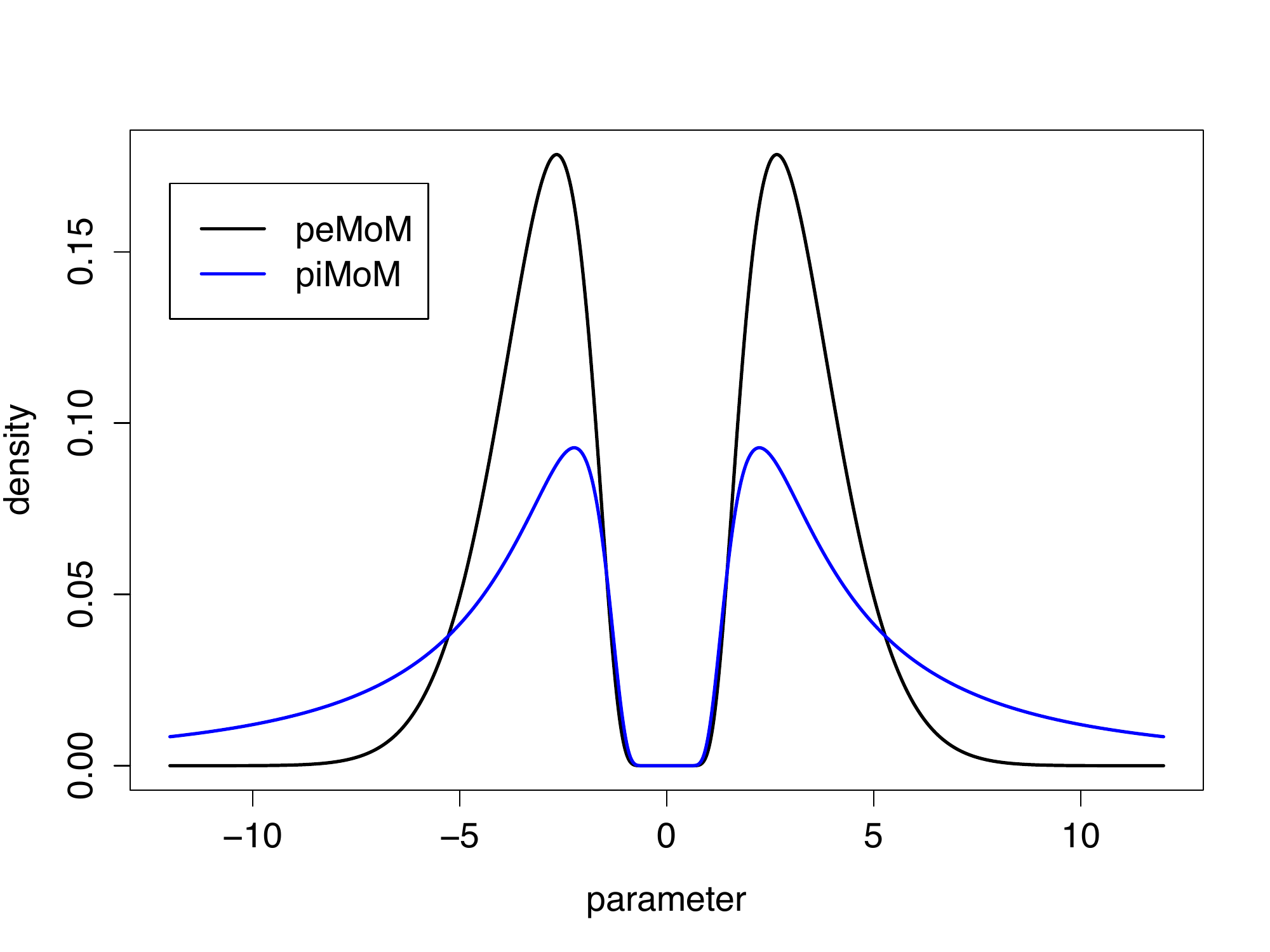}
\caption{  Nonlocal prior density functions for a single regression coefficient with $\tau=5$; for the piMoM prior, $r=1$. }
 \label{fig:nonlocal}
\end{figure}

The piMoM and peMoM prior densities are nonlocal in the sense that  the density value at the origin is exactly zero. This feature of the densities for a single regression coefficient is illustrated in Figure \ref{fig:nonlocal}. Since the piMoM prior densities and the peMoM prior densities have the same term $\exp\{-\tau/\beta^2\}$ that controls the behavior of the density function around the origin, they attain almost the same shape of the density function at the origin, which yields the same theoretical results in an asymptotic sense. Further details regarding this point are discussed in Section \ref{sec:theory}. 

We focus on these two classes of nonlocal priors in the sequel. Note that in both \eqref{eq:general_nonlocal} and \eqref{eq:mod_piMoM}, $\pi(\beta_{\bk}) = 0$ when $\beta_{\bk} = { 0}$; a defining feature of nonlocal priors. 
The distinction between the peMoM and the piMoM priors mainly involves their tail behavior. Whereas peMoM priors possess Gaussian tails, the piMoM prior densities have inverse polynomial tails. For example, piMoM densities with $r=1$ have Cauchy-like tails, which has implications for their finite sample consistency and asymptotic bias in posterior mean estimates of regression coefficients. Since similar conditions are later imposed on the hyperparameter $\tau$ appearing in \eqref{eq:general_nonlocal} and \eqref{eq:mod_piMoM}, at the risk of some ambiguity we use the same notation for the two hyperparameters in these equations.

In addition to imposing priors on the regression parameters given a model, we need to place a prior on the space of models to complete the prior specification. We consider a uniform prior on the model space restricted to models having size less than or equal to $q_n$, with $q_n < n$, i.e.,
\begin{align}\label{eq:model_pr}
\pi(\bk) \propto I(|\bk | \le q_n),
\end{align}
where $I(\cdot)$ denotes the indicator function and with a slight abuse of notation, we denote the prior on the space of models by $\pi$ as well. 
Similar priors have been considered in the literature by \citet{Jiang2007} and \citet{Liang2013}. Since the peMoM and piMoM priors already induce a strong penalty on the size of the model space (see Section \ref{sec:theory}), we do not need to additionally penalize larger models using, for example, model space priors of the type discussed in \citet{scott2010bayes}.

%
%

Under a peMoM prior \eqref{eq:general_nonlocal} on the regression coefficients, the marginal likelihood $m_{\bk}(y)$ under model $\bk$ given $\sigma^2$ can be obtained by integrating out $\beta_{\bk}$, resulting in \\
$m_{\bk}(y) 
= (2\pi \sigma^2)^{-\frac{n}{2}} \, C^{-|\bk|} \, Q_{\bk} \, \exp\{-\widetilde{R}_{\bk}/(2\sigma^2) \}$, 

where
\begin{eqnarray}
\widetilde{R}_{\bk} &=&y^T(\mr{I}_n-\widetilde{\mr{P}}_\bk)y, \quad \widetilde{\mr{P}}_\bk = X_\bk(X_\bk^\T X_\bk+1/\tau\mr{I}_{\bk})^{-1}X_\bk^\T, \label{eq:R_k}\nonumber\\ 
Q_\bk&=&\int \exp\{-(\beta_\bk-\widetilde{\beta}_\bk)^\T\widetilde{\Sigma}_\bk^{-1}(\beta_\bk-\widetilde{\beta}_\bk)/(2\sigma^2) - \sum_{j=1}^{|\bk|} \tau/\beta_{\bk,j}^2 \} d\beta_\bk, \label{eq:Q_k}\\
\widetilde{\beta}_\bk&=&(X_\bk^\T X_\bk+1/\tau\mr{I}_\bk)^{-1}X_\bk^\T y, \quad \widetilde{\Sigma}_\bk = (X_\bk^\T X_\bk+1/\tau\mr{I}_\bk)^{-1}. \label{eq:B_k}\nonumber
\end{eqnarray} 

Similarly, the marginal likelihood using the piMoM prior densities \eqref{eq:mod_piMoM} can be expressed as 
$m_{\bk}(y) = (2\pi\sigma^2)^{-\frac{n}{2}} \, C^{*-|\bk|} \, Q_{\bk}^* \, \exp\{-{R}_{\bk}^*/(2\sigma^2) \} $,
where
\begin{eqnarray}
{R}_{\bk}^* &=&y^\T(\mr{I}_n-{\mr{P}}_\bk)y, \quad {\mr{P}}_\bk = X_\bk\left(X_\bk^\T X_\bk\right)^{-1}X_\bk^\T, \label{eq:R_k2}\nonumber\\ 
Q_\bk^*&=&\int \prod_{j=1}^{|\bk|}\beta_{\bk,j}^{-2r}\exp\{-(\beta_\bk-\widehat{\beta}_\bk)^\T{\Sigma}_\bk^{*-1}(\beta_\bk-\widehat{\beta}_\bk)/(2\sigma^2) - \sum_{j=1}^{|\bk|} \tau/\beta_{\bk,j}^2 \} d\beta_\bk, \label{eq:Q_k2}\\
\widehat{\beta}_\bk&=&(X_\bk^\T X_\bk)^{-1}X_\bk^\T y, \quad {\Sigma}_\bk^* = (X_\bk^\T X_\bk)^{-1}. \label{eq:B_k2}\nonumber
\end{eqnarray} 
The integrals for $Q_{\bk}$ and $Q_{\bk}^*$ cannot be obtained in closed forms, so for computational purposes we make Laplace approximations to $m_{\bk}(y)$. { The expressions for the marginal likelihood derived here is nevertheless important for our theoretical study in Section \ref{sec:theory}.}

\section{Numerical results} \label{sec:num}
\subsection{Simulation studies using precision-recall curves}\label{sim}
To illustrate the performance of nonlocal priors in ultrahigh-dimensional settings  and to compare their performance with other methods, we calculated precision-recall curves \citep{davis2006relationship} for all selection procedures.  A precision-recall curve plots the precision = TP/(TP + FP), versus recall (or sensitivity) = TP/(TP + FN), where TP, FP and FN respectively denote the number of true positives, false positives, and false negatives, as the tuning parameter is varied. The efficacy of a procedure can be measured by the area under the precision-recall curve; the greater the area, the more accurate the method. Since both precision and recall take values in $[0, 1]$, the area under the curve for an ideal precision-recall curve is 1. We used two $(n, p)$ combinations, namely $(n,p)=(400,10000)$ and $(n,p)=(400,20000)$, and plotted the average of the precision-recall curves obtained from $100$ independent replicates of each procedure. 

We compared the performance of peMoM and piMoM priors to a number of frequentist penalized likelihood methods: lasso \citep{tibshirani1996regression}, adaptive lasso \citep{zou2006adaptive}, scad \citep{fan2001variable}, and minimax concave penalty  \citep{zhang2010nearly}. We used the \verb R ~package {\it ncvreg} to fit these penalized likelihood methods. We also included reciprocal lasso  in our simulation studies.  However, due to computational constraints involved in implementing the full rlasso procedure, we followed the recommendation in \citet{song2015high} and instead implemented the reduced rlasso.  The reduced rlasso procedure is a simplified version of rlasso that uses the least square estimators of $\beta$ when minimizing the rlasso objective function. 

We considered Zellner's $g$-prior \citep{zellner1986, liang2008mixtures} as a competing Bayesian method, with $\beta_{\bk} \mid \bk, \sigma^2 \sim N(0, g \sigma^2 (X_{\bk}^{\T} X_{\bk})^{-1})$ and $g$ is the tuning parameter. With the prior $\pi(\sigma^2)\propto 1/\sigma^2$, the marginal likelihood $m_{\bk}(y) \propto (1+g)^{-|\bk|/2}\{1+g(1-D_\bk^2)\}^{-(n-1)/2}$ can be obtained in a closed form; see for example, \citet[pp 412]{liang2008mixtures}, where $D_\bk^2$ is the ordinary coefficient of determination for the model $\bk$.

A uniform model prior \eqref{eq:model_pr} was considered for all Bayesian procedures.  This prior was chosen for several reasons.  First, construction of the PR curves requires maximization over model hyperparameters, which is most easily achieved if there is only one unknown hyperparameter.   We also wished to avoid providing an advantage to the Bayesian methods by introducing additional tuning parameters into these methods that were not present in the penalized likelihood methods.  Furthermore, the use of non-uniform priors on the model space introduces (at least) one more degree of freedom into the comparisons between methods, and our intent was to compare the effects of the penalties imposed on regression coefficients by both penalized likelihood and Bayesian methods.  At first blush, this might appear to put Bayesian methods like those based on the $g$-prior at a disadvantage, since such methods do not yield consistent variable selection even in $p<n$ settings without prior sparsity penalties on the model space (when $g$ is held fixed as $n$ increases).  However, in the construction of our PR curves, we allowed prior hyperparameters to increase with $n$, which effectively allowed the Bayesian methods to impose additional sparseness penalties through the introduction of large hyperparameter values.

We arbitrarily fixed $r = 1$ for the piMoM prior \eqref{eq:mod_piMoM} and used  an inverse-gamma prior on $\sigma^2$ with parameters $(0.1,0.1)$ for the peMoM, piMoM priors, and $g$-priors. Posterior computations for the peMoM, piMoM and $g$-priors were implemented using the Simplified Shotgun Stochastic Search with Screening (S5) algorithm described in Section \ref{sec:comp}. The maximum a posteriori model was used in each case to summarize the model selection performance. The precision-recall curves are drawn by varying the hyperparameters ($\tau$ for the nonlocal priors and $g$ for the $g$-priors), so the comparison between the model selection based on the nonlocal priors and the $g$-prior is free of the choice of hyperparameters. 
Because of their high computational burden, we could not include BASAD \citep{Narisetty2014} in the comparisons.

For each simulation setting, we simulated data according to a Gaussian linear model as in \eqref{eq:like} with the fixed true model $\bt=\{1,2,3,4,5\}$ with the true regression coefficient $\beta_\bt^0 = (0.50, 0.75, 1.00, 1.25, 1.50)^\T$ and $\sigma=1.5$. Also, the signs of the true regression coefficients were randomly determined with probability one-half.  Each row of $X$ was independently generated from a $N(0,\Sigma)$ distribution with one of the following covariance structures:
\medskip

\noindent { Case (1): } compound symmetry design; $\Sigma_{j j'}=0.5$, if $j \neq j'$ and $\Sigma_{j j}=1$, $1 \le j, j' \le p$. 

\noindent { Case (2):} autoregressive correlated design; $\Sigma_{j j'}=0.5^{|j - j'|}$, $1 \le j, j' \le p$.

\noindent { Case (3):} isotropic design; $\Sigma = \mr I_p$.

\medskip

 Figure \ref{fig:prc} plots the precision-recall curves averaged over 100 simulation replicates for the different methods across the two ($n$,$p$) pairs and the three covariate designs. From Figure \ref{fig:prc}, it is evident that the precision-recall curves for the peMoM and piMoM priors have an overall better performance than the penalized likelihood methods lasso, adaptive lasso, scad, and mcp. For decision procedures having the same power, this implies that the nonlocal priors achieve lower false discovery rates.   As discussed in Section \ref{sec:rlasso}, since the reduced rlasso shares the same nonlocal kernel as the nonlocal priors, it has a similar selection performance. The figure also shows that Zellner's $g$-prior attains comparable performance with the nonlocal priors in terms of the precision-recall curves. 
 
\begin{figure}
\centering
\includegraphics[height=21cm, width=14.5cm]{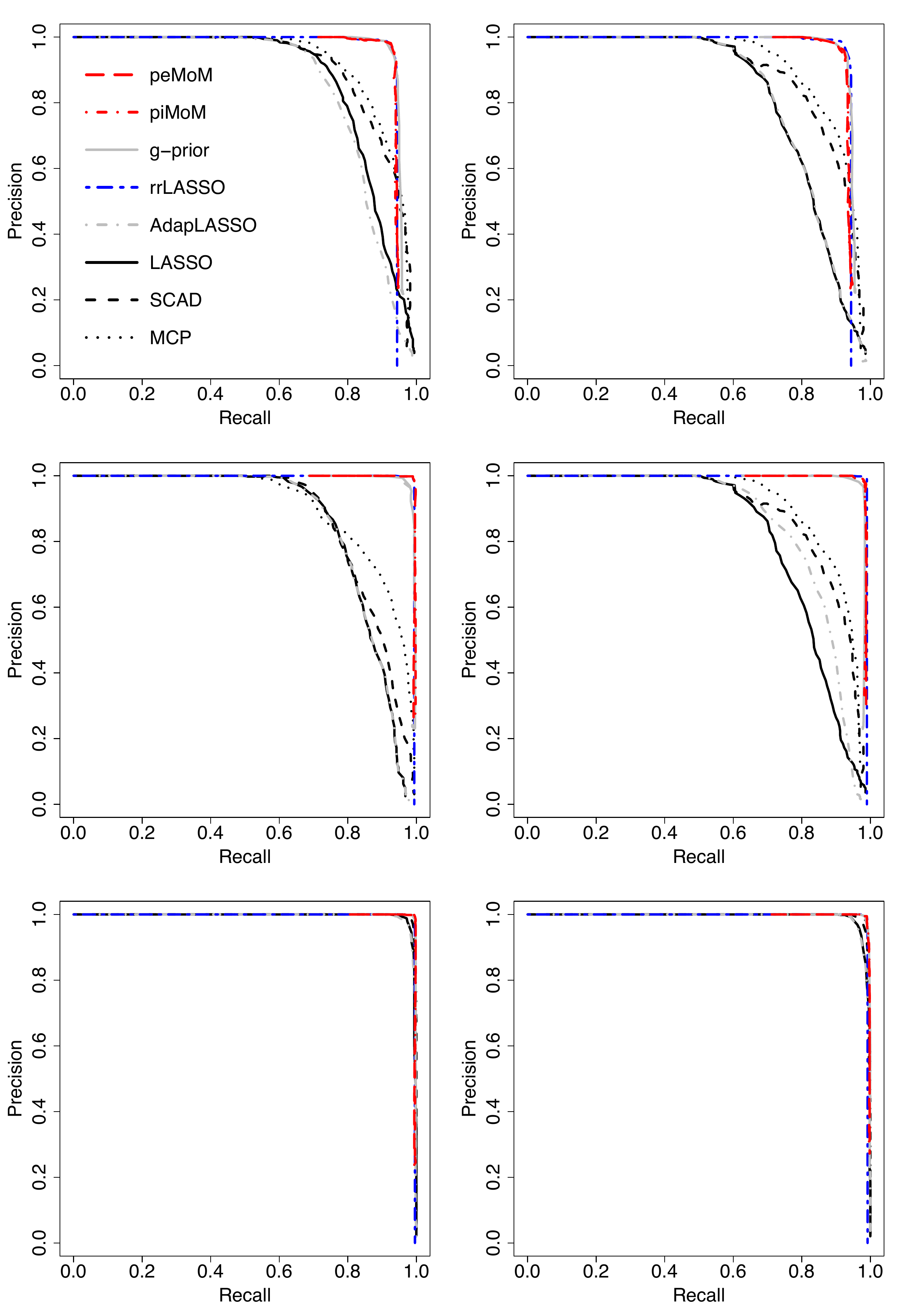}
\caption{  Plot of the mean precision-precision curves over 100 datasets with $(n,p)=(400,10000)$(first column) and $(n,p)=(400,20000)$(second column). Top: case (1); middle: case (2); bottom: case (3).}
 \label{fig:prc}
\end{figure}
\begin{figure} 
\centering
\includegraphics[height=12cm, width=14cm]{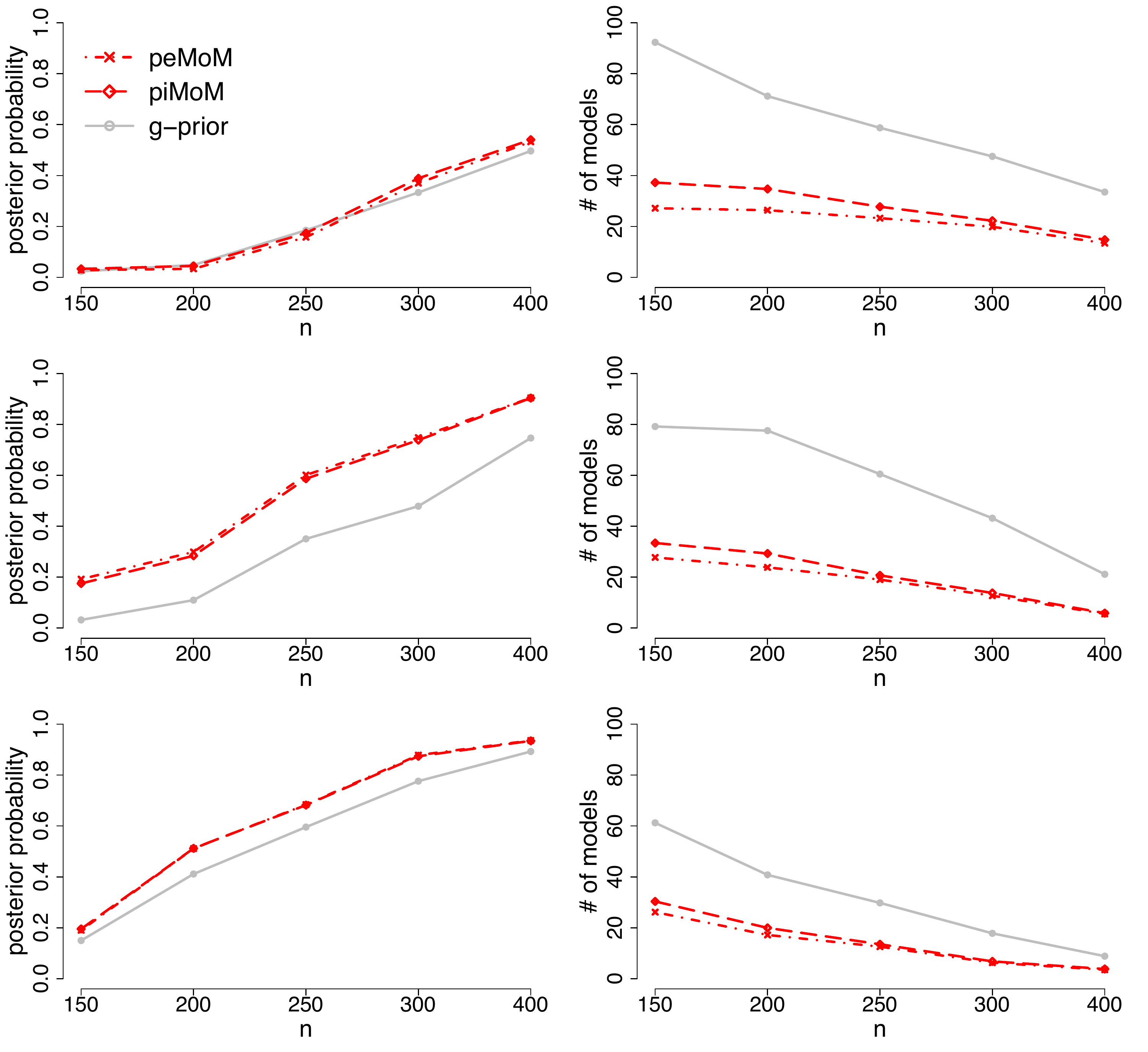}
\caption{ Averaged posterior true model probability and the number of models which attain the posterior odds ratio, with respect to the maximum a posteriori model, larger than $0.001$ with the fixed $p=20000$ and varying $n$. Top: case (1); middle: case (2); bottom: case (3).} 
\label{fig:post}
\end{figure}

\subsection{Further comparison with Zellner's $g$-prior}\label{zellner}
 
The similarity of the performances of the $g$-prior and the nonlocal priors in terms of precision-recall curves begs for closer comparisons of these procedures. For this reason, we also investigated the concentration of the posterior densities around their maximum models. To this end, we fixed $p = 20,000$ and varied $n$ from $150$ to $400$; the data generating mechanism was exactly the same as in Section \ref{sim}. The left column of Figure~\ref{fig:post} displays the posterior probability of the true model under the peMoM, piMoM and $g$-prior models versus $n$ for the three covariate designs in Section \ref{sim}. The plot shows that the posterior probability of the true model increases with $n$ for all three methods, with the peMoM and piMoM priors almost uniformly dominating the $g$-prior, implying a higher concentration of the posterior around the true model for the nonlocal priors. 

This tendency is confirmed in  the right panel of Figure \ref{fig:post}, where we plot the number of models $\bk$ which achieve a posterior odds ratio $\pi(\bk \mid y)/\pi(\widehat\bk\mid y) > 0.001$, where $\widehat\bk$ is the maximum a posteriori model. This plot clearly shows that the posterior distribution on the model space from the $g$-priors is more diffuse than those obtained using the nonlocal prior methods.{
These comparisons were based on fitting the hyperparameters $g$ and $\tau$ at their  oracle value, i.e., the value which maximized the posterior probability of the true model for a given value of $n$.  

The magnitudes of the oracle hyperparameters under each model also present an interesting contrast between the local and nonlocal priors. We observed that the oracle value of $g$ increased rapidly with $p$, whereas the oracle value of $\tau$ was much more stable. This phenomenon is illustrated in  
Table \ref{tab:hyper} that shows the oracle hyperparameter value averaged over 100 replicates for the three different covariate designs in Section \ref{sim}. For this comparison, we fixed $n = 400$ and varied $p$ between $1000$ and $20,000$; five representative values are displayed. The oracle values for $g$ are on a completely different scale from the oracle values $\tau$, and they vary more with p.  This table confirms the recommendations in \cite{george2000calibration} for setting $g = p^2$ based on minimax arguments.  However, the finite sample behavior of the optimal choice of $g$ is unclear, which means that the large variance of the optimal hyperparameter value is likely to hinder the selection of $g$ in real applications. Finally, we note that such large values of $g$ effectively convert the local $g$-priors into nonlocal priors by effectively collapsing the $g$-prior density to 0 at the origin.

\begin{table}
\caption{Optimal hyperparameters for Bayesian model selection methods}{%
\begin{tabular}{ccccccc}
 \\
 &&\multicolumn{5}{c}{The number of predictors}\\
 &  &$p=1000$&$p=2000$&$p=5000$&$p=10000$&$p=20000$\\ 
  \hline
 Case (1)& peMoM & 2.24 & 2.72 & 2.88 & 3.32  & 3.60 \\ 
             & piMoM & 2.16 & 2.59 & 2.70 & 3.04  & 3.26 \\ 
             & $g$-prior & $7.83 \times 10^8$ & $2.87 \times 10^9$ & $3.05 \times 10^9$ & $9.66 \times 10^9$  & $1. 70 \times 10^{10}$ \\ 
\hline
 Case (2)& peMoM &  $1. 97$ &  $2. 29$ &  $2. 34$ &  $2. 75$  & $3.00$ \\ 
             & piMoM &  $1.97$ & $2. 20$ & $2 .32$ &  $2.66$  & $2.86$ \\ 
             & $g$-prior & $8.56 \times 10^9$ & $2.55  \times 10^{10}$ & $2.62 \times 10^{10}$ & $6.58 \times 10^{10}$  & $1.25 \times 10^{11}$ \\ 
\hline
 Case (3)& peMoM & $2.66$ & $3.00$ & $3.00$ & $3.10$  & $3.60$ \\ 
             & piMoM & $2.61$ & $2.94$ & $2.94$ &  $2.94$ & $3.46$ \\ 
             & $g$-prior & $1.26\times 10^{12}$ & $8.84\times 10^{12}$ & $9.67\times 10^{12}$ & $6.81\times 10^{12}$  & $4.29\times 10^{13}$ \\ 
\hline
\end{tabular}}
\label{tab:hyper}
\end{table}

\section{Model selection consistency}\label{sec:theory}

{The empirical performance of the peMoM and piMoM priors suggests that the hyperparameter $\tau$ should be increased slowly with $p$. } While \cite{Johnson2012} were able to show strong selection consistency with a fixed value of $\tau$, it is not clear whether their proof can be extended to $p \gg n$ cases. Motivated by the empirical findings of the last section, we next investigated the strong consistency properties of peMoM and piMoM priors when $\tau$ was allowed to grow at a logarithmic rate in $p$. We found that in such cases, both peMoM and piMoM priors achieve strong model selection consistency under standard regularity assumptions when $p$ increases sub-exponentially with $n$, i.e., $\log p = O( n^{\alpha} )$ for $\alpha  \in (0, 1)$. 

{ Henceforth, we use $\tau_{n,p}$ instead of $\tau$ to denote the hyperparameter in the peMoM and piMoM priors  in \eqref{eq:general_nonlocal} and \eqref{eq:mod_piMoM} respectively. The normalizing constants for these priors is now denoted by $C_{n,p}$ and $C_{n,p}^*$ respectively. Before providing our theoretical results, we first state a number of regularity conditions.} Let $\nu_j(A)$ denote the $j$-th largest nonzero eigenvalue of an arbitrary matrix $A$, and let 
\begin{eqnarray}
\nu_{\bf k *}= \min_{1 \leq j \leq min(n,|\bk|)} \nu_j(X_\bk^\T X_\bk/n), \quad 
\nu_\bk^*=\max_{1 \leq j \leq min(n,|\bk|)} \nu_j(X_\bk^\T X_\bk/n). \label{eq:nukstar}
\end{eqnarray}

For sequences $a_n$ and $b_n$, $a_n\succeq b_n $ indicates $b_n=O(a_n)$, and $a_n\succ b_n $ indicates $b_n=o(a_n)$. With this notation, we assume that the following regularity conditions apply.


\bigskip
{\em Assumption 1.} There exists $\alpha \in (0, 1)$ such that $\log p = O(n^{\alpha})$. 

\bigskip
{\em Assumption 2.} $\log p \prec \tau_{n, p}\prec n$. 

\bigskip
{\em Assumption 3.} $|\bk|\leq {q_n}$, where $ {q_n}\prec \frac{\tau_{n,p}}{\log p }$.

\bigskip
{\em Assumption 4.} $\smash{\displaystyle\min_{{\bf k:|k|}\leq q_n}}\nu_{\bf k *} \succ\frac{\tau_{n, p}}{n}$.

\bigskip
{\em Assumption 5.} $C_1<\nu_{\bt*}\leq \nu_{\bt}^*<C_2$ for some positive constants $C_1$ and $C_2$.\\

Several comments regarding these conditions are worth making. Assumption 1 allows $p$ to grow sub-exponentially with $n$. Our theoretical results continue to hold when $p$ grows polynomially in $n$, i.e., at the rate $O(n^\gamma)$ for some $\gamma>1$. Assumption 2 reflects our empirical findings about the oracle $\tau \equiv \tau_{n,p}$ in Section \ref{sim}, which was observed to grow slowly with $p$. We need the bound on $q_n$ in Assumption 3 to ensure that the least square estimator of a model is consistent when a model contains the true model. In the $p \leq n$ setting, \citet{Johnson2012} assumed that all eigenvalues of the Gram matrix $(X_\bk^\T X_\bk)/n$ are bounded above and below by global constants for all $\bk$. 
However, this assumption is no longer viable when $p\gg n$ and we replace that by Assumption 4, where the minimum of the minimum eigenvalue of $(X_\bk^\T X_\bk)/n$ over all submodels $\bk$ with $|\bk| \le q_n$ is allowed to decrease with increasing $n$ and $p$.
Assumption 4 is called the sparse Riesz condition and is also used in \citet{Chen2008} and \citet{Kim2012}. \citet{Narisetty2014} showed that Assumption 4 holds with overwhelmingly large probability when the rows of the design matrix are independent with an isotropic sub-Gaussian distribution. Even though the assumption of sub-Gaussian tails on the covariates is difficult to verify, the results in \citet{Narisetty2014} show that Assumption 4 can be satisfied for some sequence of design matrices. 

 We now state a Theorem that demonstrates that model selection procedures based on the peMoM and piMoM nonlocal prior densities achieve strong consistency under the proposed regularity conditions. A proof of the Theorem is provided in the Supplemental Materials.
 \begin{theorem}\label{theo1}
 Suppose $\sigma^2$ is known and that Assumptions 1 {\textendash} 5 hold. Let $\pi(\bt\mid{\bf y})$ denote the posterior probability of the true model obtained under a peMoM prior \eqref{eq:general_nonlocal}. Also, assume a uniform prior on all models of size less than or equal to $q_n$, i.e., $\pi(\bk)\propto I(|\bk|\leq q_n)$. Then, $\pi(\bt\mid{\bf y})$ converges to one in probability as $n$ goes to $\infty$.
 \end{theorem}
\begin{corollary}\label{cor1}
Assume the conditions of the preceding Theorem apply. Let $\pi(\bt\mid{\bf y})$ denote the posterior probability of the true model obtained under a piMoM prior density \eqref{eq:mod_piMoM}.  Then, $\pi(\bt\mid{\bf y})$ converges to one in probability as $n$ goes to $\infty$.
\end{corollary}
We note that these results apply also if  a beta-Bernoulli prior is imposed on the model space  as in \citet{scott2010bayes}, because the effect of that prior is asymptotically negligible when $|\bk|\leq {q_n} \prec n$. 
 
In most applications, $\sigma^2$ is unknown, and it is thus necessary to specify a prior density on it. By imposing a proper inverse gamma prior density on $\sigma^2$, we can obtain the strong model consistency result stated in the Theorem below.  The proof is again deferred to the Supplemental Materials. 
\begin{theorem}\label{theo2}
Suppose $\sigma^2$ is unknown and a proper inverse gamma density with parameters $(a_0,b_0)$ is assumed for $\sigma^2$. Also, let $\pi(\bt\mid{\bf y})$ denote the posterior probability of the true model evaluated using peMoM priors. Then if Assumptions 1 {\textendash} 5 are satisfied, $\pi(\bt\mid{\bf y})$ converges to one in probability as $n$ goes to $\infty$.
\end{theorem}

\begin{corollary}\label{cor2}
Suppose the conditions of the preceding Theorem apply, but that $\pi(\bt\mid{\bf y})$ now denotes the posterior probability of the true model obtained under a piMoM prior density. Then $\pi(\bt\mid{\bf y})$ converges to one in probability as $n$ goes to $\infty$.
\end{corollary}

\section{Connections between nonlocal priors and reciprocal lasso}\label{sec:rlasso}

In this section, we highlight the connection between the rlasso of \citet{song2015high} and Bayesian variable selection procedures based on our nonlocal priors.  We begin by noting that the objective function $g(\beta_\bk ; \bk)$ of rlasso on a model {\bk}  can be expressed as follows:
\begin{eqnarray}\label{eq:rLASSO} 
g( \beta_\bk ; \bk) = \norm y - X_\bk \beta_\bk\norm_2^2 +\sum_{j=1}^{|\bk|} \tau_{n,p}/|\beta_{\bk,j}| .
\end{eqnarray}
{The optimal model is selected by minimizing this objective function with respect to $\beta_\bk$ and \bk. It is clear that the penalty function $\sum_{j=1}^{|\bk|}\tau_{n,p}/|\beta_{\bk,j}| $ in \eqref{eq:rLASSO} is similar to the negative log-density of piMoM nonlocal priors as proposed in \citet[pp 659]{Johnson2012} and \citet[pp 149]{johnson2010use}.  The main difference between the nonlocal prior version of rlasso and the piMoM-type prior densities proposed in the previous section is the power of $\beta$ in the exponential kernels.  For the rlasso penalty this power is 1, while for piMoM-type prior densities it is 2.   The implications of this difference are apparent from the following proposition.
  }
 \begin{proposition}\label{lemma:rLASSO}
 For a given model $\bk$, suppose that $\widetilde\beta_\bk^*$ is the minimizer of the objective function \eqref{eq:rLASSO}, and again let $\widehat\beta_\bk$ denote the least square estimator of $\beta$ under model $\bk$. Assume that $\tau_{n,p}\prec n$, and there exist strictly  positive contants $C_L$ and $C_U$ such that $C_L<\nu_{\bk*} \leq \nu_{\bk}^*<C_U$. Then, for any $\epsilon_n^*\succ (\tau_{n,p}/n)^{1/3}$, 
 \[
 P\left[ \widetilde\beta_\bk^* \notin {R}\big(\widehat\beta_\bk; \epsilon_n^* \big)\right] \rightarrow 0,
 \]
 where  ${R}({u} ; \epsilon)=\{ {\bf x} \in \mathbb{R}^{|\bk|}: |x_j-u_j| \leq \epsilon, j=1,\dots,|\bk| \}$. 
 \end{proposition}

{The proposition shows that under standard conditions on the eigenvalues of the Gram matrix $X_\bk^TX_\bk/n$, the estimator derived from \eqref{eq:rLASSO} is asymptotically within $(\tau_{n,p}/n)^{1/3}$ distance of the least squares estimator $\widehat\beta_\bk$. 
On the other hand, results cited in the previous section show that maximum a posteriori  estimators obtained from the piMoM-type prior densities reside at an asymptotic distance of $(\tau_{n,p}/n)^{1/4}$ from the least squares estimator.  Variable selection procedures based on both forms of piMoM priors thus achieve adaptive penalties on the regression coefficients in the sense described in \citet{song2015high}.
 }
 
  Although rlasso is proposed as a penalized likelihood approach, the computational procedure  to optimize its objective function is quite different from the other penalized likelihood methods. The resulting computational complexity of this optimization procedure, which contains a discontinuous penalty function, is NP-hard.  This suggests that the formulation of this nonlocal penalty in a penalized likelihood framework is unlikely to provide significant computational advantages over related Bayesian model selection procedures, even though the inferential advantages of the Bayesian framework are lost.  
\section{Asymptotic behavior of marginal likelihoods based on nonlocal priors}
From Lemma \ref{LemmaQ} in the Supplemental Materials, it follows that the asymptotic log-marginal likelihood of a model {\bk} based on a peMoM or piMoM prior density can be expressed as
\begin{eqnarray*}
\log\pi(\bk\mid y)&=&l(\widehat\beta_{\bk}) +\log Q_{\bk}-|\bk|\log C_{n,p}\\
&\approx& l(\widehat{\beta}_{\bk}) - \sum_{j=1}^{|\bk|} p_{\tau_{n,p}}\big(\widehat\beta_{\bk,j}\big)+C, 
\end{eqnarray*}
for some constant $C$, $\widehat{\beta}_{\bk}$ is the maximum likelihood  estimator under the model $\bk$,  {\it i.e.} $\widehat{\beta}_{\bk}=(X_{\bk}^TX_{\bk})^{-1}X_{\bk}^T y$, and
\begin{eqnarray}\label{eq:penalty}
p_{\tau_{n, p}}\big(\widehat{\beta}_{\bk,j}\big) \approx \begin{cases} 
(n\tau_{n, p}u_\bk)^{1/2}, &\mbox{if }|\widehat{\beta}_{\bk,j}| < c\big( \frac{nu_\bk}{\tau_{n, p}}\big)^{-1/4}\\
\tau_{n, p}/\widehat{\beta}_{\bk,j}^2, &\mbox{if }|\widehat{\beta}_{\bk,j}| \geq c\big( \frac{nu_\bk}{\tau_{n, p}}\big)^{-1/4},
\end{cases}
\end{eqnarray}
 for some constant $c$ and some arbitrary sequence $u_\bk$ with $\nu_{\bk*}\leq u_\bk\leq \nu_{\bk}^*$. We note that the strength of the correlation between the variables in the model $\bk$ affects the behavior of $u_\bk$, and $(nu_\bk/\tau_{n,p})^{-1/4}$ converges to zero as $n$ tends to infinity due to Assumption 4 described in Section  \ref{sec:theory}.

 On the other hand, the penalty term in the other Bayesian model selection approaches is quite different from that of the nonlocal priors as in \eqref{eq:penalty}. The marginal likelihood based on the $g$-prior when $\sigma^2$ is known can be expressed as 
  \[
l(\widehat{\beta}_\bk) - |\bk|\log (1+g)/2.
\]

Narisetty and He(2014) demonstrated that BASAD achieves the  strong model selection consistency.  This consistency follows from that the fact that the BASAD ``penalty'' is asymptotically equivalent to
\begin{eqnarray}\label{eq:basad}
l(\widehat{\beta}_\bk) - c|\bk|\log (p),
\end{eqnarray}
where $c$ is some constant. \cite{yang2015computational} and \cite{castillo2012needles} also considered a similar penalty term on the model space, which implies that the posterior probability for their procedures can be expressed in the same form as  \eqref{eq:basad}. When $g = p^{2c}$, the marginal likelihood based on a $g$-prior is asymptotically equivalent to   \eqref{eq:basad}.

The asymptotic term of the marginal likelihoods is quite different from that of the  nonlocal priors, since the penalty terms in the other Bayesian approaches only focus on the model size without considering the different weights on variables in the model. The marginal likelihoods based on nonlocal priors, however, impose different penalties on each predictor in the given model. When the MLE of the regression coefficient in the model is asymptotically close to zero ($|\widehat{\beta}_{\bk,j}|< c( nu_\bk/\tau_{n, p})^{-1/4}$), the model that contains the corresponding variable would be strongly penalized by $(n\tau_{n,p}u_\bk)^{1/2}$. In contrast, when the MLE is asymptotically significant ($|\widehat{\beta}_{\bk,j}|\geq c( nu_\bk/\tau_{n, p})^{-1/4}$), the penalty attains a different weight based on the MLE ($p_{\tau_{n, p}}(\widehat{\beta}_{\bk,j}) \approx \tau_{n, p}/\widehat{\beta}_{\bk,j}^2$).  

This analysis highlights the fact that the nonlocal priors are able to adapt their penalty for the inclusion of covariates based on the observed data, whereas the local priors must instead rely on a prior penalty on non-sparse models.  

\section{Computational strategy}\label{sec:comp}
In $p \gg n$ settings, full posterior sampling using existing Markov chain Monte Carlo (MCMC)  algorithms is highly inefficient and often not feasible from a practical perspective.  To overcome this problem, we propose a scalable stochastic search algorithm aimed at rapidly identifying regions of high posterior probability and finding the maximum a posteriori (MAP)  model. Our main innovation is to develop a stochastic search algorithm combining isis-like screening techniques \citep{fan2008sure} and temperature control that is commonly used in global optimization algorithms such as simulated annealing \citep{kirkpatrick1983optimization}. 

To describe our proposed algorithm, consider  the MAP model $\widehat{\bk}$ that can be expressed as
\begin{eqnarray}\label{MAP}
\widehat{\bk} =\underset{\bk\in\Gamma^*} {\mathrm{argmax}} \{\pi(\bk\mid y)\},
\end{eqnarray} 
where $\Gamma^*$ is the set of all models assigned non-zero prior probability.  

\subsection{Shotgun stochastic search algorithm}
\citet{hans2007shotgun} proposed the  shotgun stochastic search (SSS)  algorithm in an attempt to efficiently navigate through very large model spaces and identify global maxima. Letting $\mbox{nbd}(\bk)=\{ \Gamma^+,\Gamma^-,\Gamma^0 \}$, where $\Gamma^+= \{ \bk\cup \{j\} : j\in \bk^c \}$, $\Gamma^-= \{ \bk \setminus \{j\} : j \in \bk \}$, and $\Gamma^0=\{ [\bk \setminus \{j \}] \cup \{l\} : l \in \bk^c, j \in \bk \}$, the SSS procedure is described in {\bf Algorithm 1}.

\begin{algorithm}[!h]
\caption{Shotgun Stochastic Search (SSS)}

\bigskip
\vspace*{-12pt}
\begin{tabbing}
\enspace Choose an initial model $\bk^{(1)}$\\
\enspace For $i=1$ to $i=N-1$\\
\qquad Compute $\pi(\bk\mid y)$ for all $\bk \in \mbox{\mbox{nbd}}(\bk^{(i)})$\\
\qquad Sample $\bk^+$, $\bk^-$, and $\bk^0$, from $\Gamma^+$, $\Gamma^-$, and $\Gamma^0$, with probabilities proportional to $\pi(\bk\mid y)$\\
\qquad Sample $\bk^{(i+1)}$ from $\{\bk^+, \bk^-, \bk^0\}$, with probability proportional to\\
 \qquad \enspace$\{\pi(\bk^+\mid  y), \pi(\bk^-\mid y),\pi(\bk^0\mid y)\}$
\end{tabbing}
\end{algorithm}
The MAP model can be identified by the model that achieves the largest (unnormalized) posterior probability among those models searched by SSS.

\subsection{Simplified shotgun stochastic search algorithm with screening (S5)}
SSS is effective in exploring regions of high posterior model probability, but its computational cost is still expensive because it requires the evaluation of marginal probabilities for models in $\Gamma^+$, $\Gamma^-$, and $\Gamma^0$ at each iteration. The largest computational burden occurs for the evaluation of marginal likelihood for models in $\Gamma^0$, since $|\Gamma^0| = |\bk|(p-|\bk|)$. 
To improve the computational efficiency of SSS, we propose a modified version which only examines models in $\Gamma^+$ and $\Gamma^-$, which have cardinality $p-|\bk|$ and $|\bk|$, respectively. However, by ignoring $\Gamma^0$ in the sampling updates we make the algorithm less likely to explore ``interesting'' regions of high posterior model probability, and therefore more likely to get stuck in local maxima. To counter this problem, we introduce a ``temperature parameter'' analogous to simulated annealing which allows our algorithm to explore a broader spectrum of models. 

Even though ignoring models in $\Gamma^0$ reduces the computational burden of the SSS algorithm, the calculation of $p$ posterior model probabilities in every iteration is still computationally prohibitive when $p$ is very large. To further reduce the computational burden, we borrow ideas from the  Iterative Sure Independence Screening (isis; \cite{fan2008sure}) and consider only those variables which have a large correlation with the residuals of the current model. More precisely, we examine the products $|r_{\bk}^TX_j|$, where $r_\bk$ is the residual of the model $\bk$, for $j=1,\ldots,p$, after every iteration of the modified shotgun stochastic search algorithm, and then restrict attention to variables for which $\{|r_{\bk}^TX_j|:j=1,\ldots,p\}$ is large (we assume that the columns of ${ X}$ have been standardized). This yields a scalable algorithm even when the number of variables $p$ is large.

With these ingredients, we propose a new stochastic model search algorithm called Simplified Shotgun Stochastic Search with Screening (S5), which is described in {\bf Algorithm 2}.
\begin{algorithm}[!h]
\vspace*{-12pt}
\caption{ Simplified Shotgun Stochastic Search with Screening (S5)}
\bigskip
\begin{tabbing}
\enspace Set a temperature schedule $t_1>t_2>\ldots>t_L>0$\\
\enspace Choose an initial model $\bk^{(1,1)}$ and a set of variables after screening ${\bf S}_{\bk^{(1,1)}}$ based on $\bk^{(1,1)}$\\

\enspace For $l=1$ in $l=L$\\

\qquad For $i$ in $1,\ldots,J-1$\\

\qquad\enspace Compute all $\pi(\bk\mid y)$ for all $\bk \in \mbox{nbd}_{scr}(\bk^{(i,l)})$\\
\qquad\enspace Sample $\bk^+$ and $\bk^-$, from $\Gamma_{scr}^+$ and $\Gamma^-$, with probabilities proportional to $\pi(\bk\mid y)^{1/t_l}$\\
\qquad\enspace Sample $\bk^{(i+1,l)}$ from $\{\bk^+, \bk^-\}$, with probability proportional to $\{\pi(\bk^+\mid y)^{1/t_l}, \pi(\bk^-\mid y)^{1/t_l}\}$\\
\qquad\enspace Update the set of considered variables ${\bf S}_{{\bk}^{(i+1,l)}}$ to be the union of variables in ${\bk}^{(i+1,l)}$ and \\
\qquad\enspace the top $M_n$ variables according to $\{ |r_{\bk^{(i+1,l)}}^TX_j|:j=1,\ldots,p \}$
\end{tabbing}
\end{algorithm}

In S5, ${\bf S}_\bk$ is the union of variables in ${\bk}$ and the top $M_n$ variables, obtained by screening using the residuals from the model \bk. The screened neighborhood of the model {\bk} can be defined as $\mbox{nbd}_{scr}(\bk)=\{ \Gamma_{scr}^+,\Gamma^-\}$, where $\Gamma_{scr}^+= \{ \bk\cup \{j\} : j\in \bk^c\cap {\bf S}_\bk \}$. 

Even though this algorithm is designed to identify the MAP model, it also provides an approximation to the posterior model probability of each model.  The uncertainty of the model space can be measured by approximating the normalizing constant from the  (unnormalized) posterior probabilities of the models explored by the algorithm. 


Denoting  the computational complexity of the evaluation of the unnormalized posterior model probability of the largest model among searched models by $E_n$, the computational complexity of the SSS algorithm can be expressed as the product of the number of explored models by the algorithm and $E_n$, which is  $[O\{Np\}+O\{Nq_n\}+O\{N(p-q_n)q_n\}]\times E_n$, where $q_n$ is the maximum size of model among searched models and $q_n < n \ll p$. 

 On the other hand,  the S5  only considers  $M_n$ variables after the screening step in each iteration, which dramatically reduces the number of models to be considered in constructing the neighborhood, $O\{JL(M_n-q_n)\}+O(JLM_n)$. Therefore, the resulting computational complexity is
\be
\left[O\{JL(M_n-q_n)\}+O(JLM_n)\right] \times E_n + O(JLnp),
\ee  
where $q_n<M_n$. When the computational complexity for screening steps, $O(JLnp)$, is dominated by the other terms, the computational complexity is almost independent of $p$.  As a result, the proposed algorithm is scalable in the sense that the resulting computational complexity is typically robust to the size of $p$.  
    
\subsection{Performance comparisons between S5 and SSS}
We examined the computational performance of S5 to SSS in identifying the MAP model under a piMOM prior with $\tau_{n,p}=\log n \log p$ and $r=1$. We generated data according to Case (1) in Section \ref{sec:num}, with a fixed sample size ($n=200$), and a varying number of covariates $p$. We set $M_n = 20$, $L=20$, and $J=20$ for S5. To match the total number of iterations between S5 and SSS, we set $N=400$ for SSS. All computations were implemented in \texttt{R}. 
   
Figure \ref{fig:s5} shows the average computation time and the number of models searched before hitting the MAP model for the first time for the S5 and SSS algorithms.  All averages were based on 100 simulated datasets, and both algorithms obtained the same MAP model for all data sets.  Panel (a) shows that the computation time of SSS increases roughly at a $p^2$ rate, but that the computation time for S5 was nearly independent of the number of covariates $p$ (about 4 seconds). For example when $p=2,000$, SSS first found the MAP model in an average of 1,360 seconds (about 23 minutes), whereas S5 hit the MAP model after about only 4 seconds.  Interestingly, panel (b) of Figure \ref{fig:s5} also illustrates that the S5 algorithms explored only 181 models on average to hit the MAP model, whereas SSS typically visited slightly more than 38,000 models. Thus, not only is S5 much faster than SSS in identifying the MAP model, but it also visited far fewer models before visiting the MAP model. 
   
 \begin{figure}
    \centering
    \begin{subfigure}[b]{0.49\textwidth}
        \includegraphics[width=\textwidth]{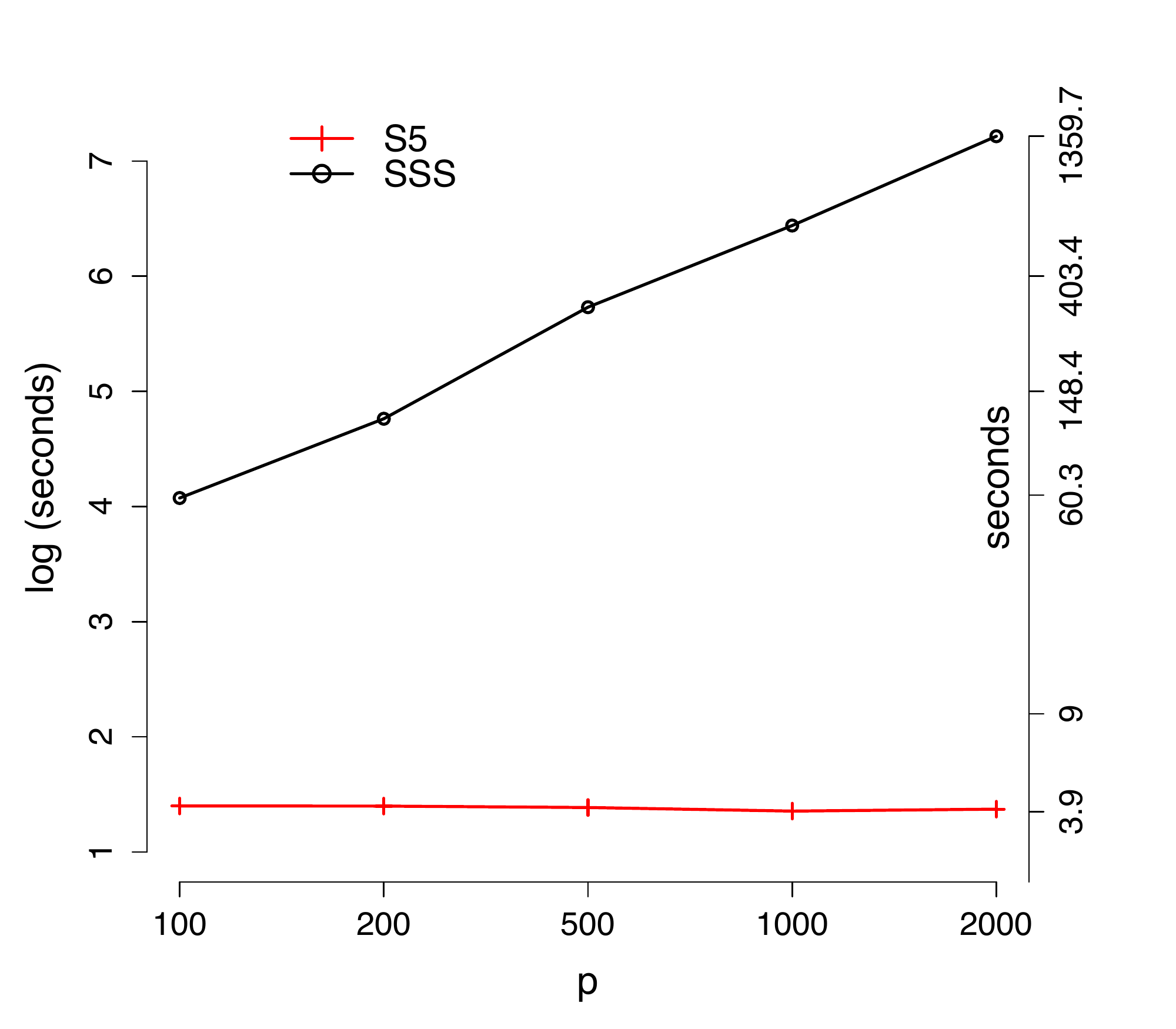}
        \caption{}
        \label{fig:gull}
    \end{subfigure}
    \begin{subfigure}[b]{0.49\textwidth}
        \includegraphics[width=\textwidth]{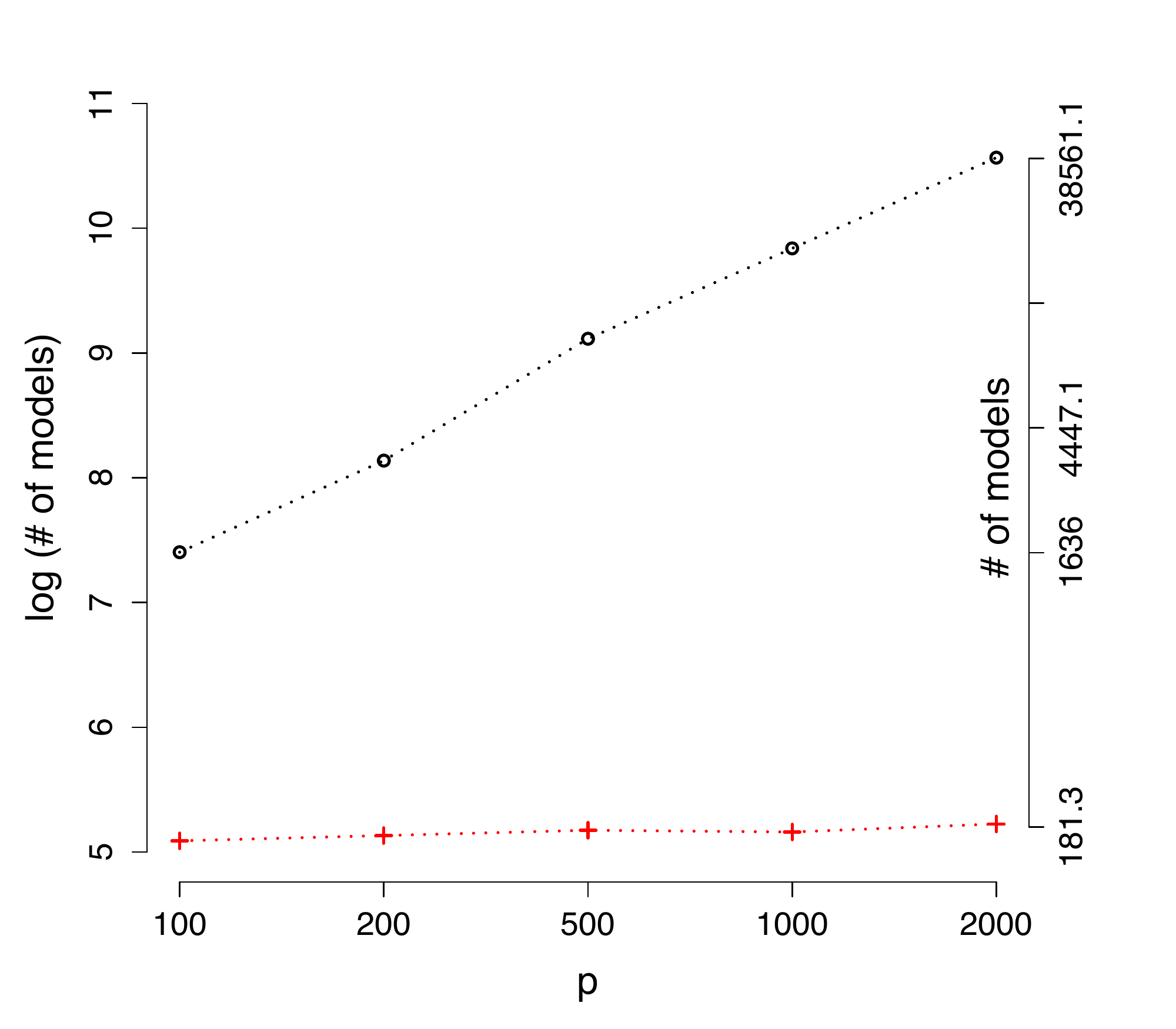}
        \caption{}
        \label{fig:mouse}
    \end{subfigure}
    \caption{(a) the average computation time to first hit the MAP model; (b) the average number of models searched before hitting the MAP model.  The left $y$-axis is in a logarithmic scale and the right $y$-axis is in the raw scale.}\label{fig:s5}
\end{figure}

\section{Real data analysis}
\subsection{Analysis of polymerase chain reaction (PCR) data}
\cite{lan2006combined} studied coordinated regulation of gene expression levels on 31 female and 29 male mice ($n=60$). 
A number of psychological phenotypes, including numbers of stearoyl-CoA desaturase 1 (SCD1), glycerol-3-phosphate acyltransferase (GPAT) and phos- phoenopyruvate carboxykinase (PEPCK), were measured by quantitative real-time RT-PCR, along with 22,575 gene expression values. The resulting data set is publicly available at \url{http://www.ncbi.nlm.nih.gov/geo} (accession number GSE3330). 

\cite{zhang2009penalized} used penalized orthogonal components regression to predict the three phenotypes mentioned above based on the high-dimensional gene expression data. \cite{bondell2012consistent} also used the same data set to examine their model selection procedure based on penalizing regression coefficients within a (marginal or joint) credible interval obtained from a ridge-type prior. For brevity, we restrict attention here to SCD1 as the response variable. 

 
Since the ground truth regarding the true significant variables is not known for this data, we compared our approach with a host of competitors on predictive accuracy and parsimony of the selected model. 

Prior to analyses, we standardized the covariates and randomly split the data set into 5 test samples and 55 training samples to evaluate the out-of-sample mean square prediction  error (MSPE)  
\[
\mbox{MSPE} = \sum_{i \in T_{test}} (y_i - X_i^T\widehat\beta_{\widehat\bk}^{tr})^2/|T_{test}|,
\]
where $T_{test}$ is the index set of the test samples and $\widehat\beta_{\widehat\bk}^{tr}$ is the least square estimator under the estimated model $\widehat \bk$ based on the training samples. To avoid sensitivity to a particular split, we considered 100 replications of the training and test sample generation. 
To measure the stability of model selection, we considered the number of variables that were (i) selected at least 95 times, and (ii) at least once, out of the 100 replicates. 

Due to the high-computational burden of the penalized credible interval \citep{bondell2012consistent} approach, we followed the pre-processing step suggested in their article to marginally screen variables to reduce to 2000 variables (1999 genes and gender). For all the other approaches, all 22,575 genes were used. For the nonlocal priors, we considered both the MAP estimator and the least squares (LS) estimator from the MAP model. For the $g$-prior, we set $g=p^2$ as recommended in \cite{george2000calibration}. For the penalized likelihood procedures, we used ten-fold cross validation to choose the tuning parameter. 

{ To choose the hyperparameter $\tau_{n,p}$ for the nonlocal priors, we used a procedure proposed by \cite{nikooienejad2016bayesian}. That procedure sets the hyperparameter so that the $L_1$ distance between the posterior distribution on the regression parameters under the null distribution (i.e., $\beta = 0$) and the nonlocal prior distributions on these parameters is constrained to be less than a specified value (e.g., $p^{-1/2}$). The average value of the hyperparameter values chosen by this procedure were  $\tau_{n,p}= 1.12$ and $\tau_{n,p}= 1.16$ for piMoM and peMoM priors, respectively. }



To make the comparison between the nonlocal priors and the $g$-prior more transparent, we used the same beta-binomial prior on the model space in both models, rather than the uniform prior on the model space described previously.  The form of the beta-binomial prior was given by 

\begin{eqnarray}\label{eq:model}
\pi({\bf k}) \propto \rho^{|{\bf k }|}(1-\rho)^{p-{|{\bf k }|}}I(|{\bf k}|\leq q_n),
\end{eqnarray}
with a uniform prior on $\rho$ and $q_n=40$.  We note that this prior does not strongly induce sparsity as does, for example, the prior obtained by imposing a $Beta(1,p^u)$, $u>1$ prior on $\rho$, as suggested in \cite{castillo2015bayesian}.

\begin{table}[ht]
\centering
\begin{tabular}{c|cccc}%
  \hline
  \hline
Method & MSPE &  MS & FS & TS  \\ 
  \hline

piMoM(MAP) & 0.283 (0.17)& 1.00 (0.00)& 1 & 1  \\ 
  piMoM(LS) & {\bf 0.282} (0.17)& 1.00 (0.00)& 1 & 1  \\ 
  peMoM(MAP) & 0.291 (0.18)& 1.02 (0.14)& 1 & 2  \\ 
  peMoM(LS) & 0.287 (0.17)& 1.02 (0.14)& 1 & 2  \\ 
  g-prior & 0.368 (0.20)& 4.07 (0.56)& 1 & 133  \\ 
  lasso & 0.542 (0.39) & 17.97 (8.62)& 1 & 211  \\ 
  scad & 0.308 (0.23)& 12.66 (7.62)& 2 & 163 \\ 
  mcp & 0.308 (0.21)& 2.20 (0.94)& 0 & 29 \\ 
  Marginal($p=2000$) & 0.456 (0.40)& 17.47 (11.16)& 0 & 273 \\ 
  Joint($p=2000$)& 0.440 (0.40)& 16.42 (11.06) & 1 & 185  \\ 
\hline
\end{tabular}
\caption{Analysis of the PCR data. Marginal and Joint refer to the variable selection procedures \citep{bondell2012consistent} based on Bayesian marginal credible set and Bayesian joint credible set, respectively. MS is the average size of the selected model. FS is the number of frequently selected variables, i.e., that were selected at least 95 times in 100 repetitions. TS refers to the total number of  variables selected at least once from 100 repetitions. Standard errors are provided in parenthesis.}
\label{tab:gene}
\end{table}

Table \ref{tab:gene} summarizes the results from the analysis of the gene expression data set. 
On average, the nonlocal priors simultaneously produced the lowest MSPE and the most parsimonious model.  
The other model selection methods selected a wide array of  different variables for different splits of the data set. In particular, lasso and the penalized credible region approach selected more than 180 different variables from 100 repeated splits, while the average size of the selected model was less than 20 and the number of frequently selected variables was only zero or one, indicating a potentially large number of false positives picked up by these methods.


\subsection{A simulation study based on the Boston housing data}
 We next examined the Boston housing data set that contains the median value of owner-occupied homes in the Boston area, together with several variables that might be associated with their median value. There were $n=506$ median values in the data set, and we considered 10 continuous variables as the predictor variables: \texttt{crim}, \texttt{indus}, \texttt{nox}, \texttt{rm}, \texttt{age}, \texttt{dis}, \texttt{tax}, \texttt{ptratio}, \texttt{b}, and \texttt{lstat}. This data set has been used to validate a variety of approaches; some recent examples relevant to variable selection include \cite{radchenko2011improved}, \cite{yuan2012efficient}, and \cite{rovckova2014emvs}.
 
To examine the model selection performance in high-dimensional settings, we added 1,000 noise variables that were generated independently from a standard Gaussian distribution ($p=1,010$). The same competitors from the previous subsection were used with the aforementioned choice of hyperparameters. For  nonlocal priors, the hyperparameter was chosen by the aforementioned procedure \citep{nikooienejad2016bayesian}; the average of the chosen hyperparameter values were $\tau_{n,p}= 2.01$ and $\tau_{n,p}= 0.47$ for piMoM and peMoM priors, respectively.} Prior to analyses, we standardized the covariates and considered a simulation test size of 100 samples. 

\begin{table}[ht]
\centering
\begin{tabular}{cccccc}
  \hline
Methods & MSPE & MS-O & MS-N & FS-O & TS-O \\ 
  \hline
piMoM(MAP) & 24.281 (9.01) & 5.05 (0.22)& 0.01 (0.10)& 5 & 6 \\ 
  piMoM(LS) & 24.265 (9.04) & 5.05 (0.22)& 0.01 (0.10)& 5 & 6 \\ 
  peMoM(MAP) & {\bf 24.156} (9.02) & 5.02 (0.14)& {\bf 0.00} (0.00) & 5 & 6 \\ 
  peMoM(LS) & 24.165 (9.00)& 5.02 (0.14)& {\bf 0.00} (0.00)& 5 & 6 \\ 
  g-prior & 26.314 (9.87)& 3.10 (0.44)& {\bf 0.00} (0.00)& 3 & 5 \\ 
  lasso & 30.243 (11.82) & 5.07 (0.87)& 7.77 (11.16)& 4 & 8 \\ 
  scad & 33.993 (10.66) & 5.39 (0.57) & 31.60 (28.28)& 5 & 7 \\ 
  mcp & 26.191 (9.87) & 4.66 (0.74) & 0.54 (1.04)& 3 & 6 \\ 
  Marginal & 26.612 (10.16) & 3.74 (0.88) & 0.41 (0.72)& 3 & 7 \\ 
  Joint & 26.385 (10.25) &  3.77 (0.94)& 0.02 (0.20)& 3 & 6 \\ 
   \hline
\end{tabular}
\caption{The Boston Housing data set:  MS-O and MS-N refer to the average number of selected original variables and selected noise variables, respectively. FS-O is the number of  original variables that are frequently selected at least 95 times out of 100 repetitions. TS-O refers to the number of original variables selected at least once from 100 repetitions.}
\label{tab:boston}
\end{table}
The results of are analysis are summarized in Table \ref{tab:boston}. The conclusions are similar to those reported in Section 8.1; the nonlocal priors consistently choose more parsimonious models and had better predictive performance.  The model selection procedure resulting from the nonlocal prior selects almost the same variables across the 100 repetitions. The average number of the original variables selected more than 95 times over 100 repetitions is 5, which is close to the average model size. It is also reliable in the sense that the average number of the original variables that are selected at least once across the repetitions is only 6. This means that model selection based on the nonlocal prior selects the same model in most data splits. On the other hand, penalized likelihood methods such as lasso and scad tend to select a large number of noise variables.


\section{Conclusion}
This article describes theoretical properties of peMoM and piMoM priors for variable selection in ultrahigh-dimensional linear model settings.  In terms of identifying a ``true'' model, selection procedures based on  peMoM priors are asymptotically equivalent to piMoM priors in \citet{Johnson2012} because they share the same kernel, $\exp\{ -\tau_{n, p}/\beta^2\}$.  We demonstrated that model selection procedures based on peMoM priors and piMoM priors achieve strong model selection consistency in $p\gg n$ settings. 

{In Section \ref{sim}, precision-recall curves were used to show that the model selection procedure based on a $g$-prior can achieve nearly the same performance in identifying the MAP model as nonlocal priors when an optimal value for the hyperparameter $g$ is chosen. However, as shown in Section \ref{zellner}, the value of the hyperparameter that maximizes the posterior probability of the true model is very large and has high variability, which may limit the practical application of this method. To overcome this problem, one can consider mixtures of $g$-prior as in \cite{liang2008mixtures}, but the asymptotic behavior of Bayes factor and model selection consistency in ultrahigh-dimensional settings have not been examined for hyper-$g$ priors, and they are difficult to implement computationally.}

In Section  \ref{sec:comp}, we proposed an efficient and scalable model selection algorithm called S5. By incorporating the SSS with a screening idea and a temperature control, S5 was able to accelerate the computation speed without losing the capacity to explore the interesting region in the model space.  Under some simulation settings, it outperformed the SSS in a sense that not only did S5 search the MAP model much faster than the SSS, but it also found exactly the same MAP model that was searched by the SSS.  

  Because the explicit form of the marginal likelihood of the nonlocal priors is not available, we used the Laplace approximation throughout the paper, and \cite{barber2016laplace} studied  the accuracy of the approximation in Bayesian high-dimensional variable selection, especially when the dimension of the approximation (which is $q_n$) and $n$ are both increasing. However, their results do not apply to the case of the nonlocal priors, since the nonlocal priors violate their regularity condition (nonzero density at the origin).  While empirical results in this paper and \cite{Johnson2012} suggest that the use of the Laplace approximation is reasonable, in future work it is still worth paying attention to the approximation error of the Laplace approximation to the marginal likelihood of the nonlocal priors.  

The close connection between our methods and the reduced rlasso procedures provides a useful contrast between Bayesian and penalized likelihood methods for variable selection procedures.  According to the evaluation criteria proposed in Section \ref{sec:rlasso}, the two classes of methods appear to perform quite similarly.  A potential advantage of the reduced rlasso procedure, and to the lesser extent the rlasso procedure, is reduced computation cost.  This advantage accrues primarily because the reduced rlasso can be computed from the least squares estimate of each model's regression parameter, whereas the Bayesian procedures require numerical optimization to obtain the maximum a posteriori  estimate used in the evaluation of the Laplace approximation to the marginal density of each model visited.  However, the procedures used to search the model space, given the value of a marginal density or objective function, are approximately equally complex for both classes of procedures.  There are also potential advantages of the Bayesian methods.  For example, it is possible to approximate the normalizing constant of the posterior model probability from the models visited by S5 algorithm, and to use this normalizing constant to obtain an approximation to the posterior probability assigned to each model.  In so doing, the Bayesian procedures provide a natural estimate of uncertainty associated with model selection.  These posterior model probabilities can also be used in Bayesian modeling averaging procedures, which have been demonstrated to improve prediction accuracy (e.g., \cite{raftery1997bayesian}) over prediction procedures based on maximum a posteriori  estimates.  Finally, the availability of prior densities may prove useful in setting model hyperparameters (i.e., $\tau_{n,p}$) in actual applications, where scientific knowledge is typically available to guide the definition of the magnitude of substantively important regression parameters.  
 
 We also developed an \verb R ~package {\it BayesS5} that provides all computational functions used in this paper, including a support of parallel computing environments. It is available on the author's website and on \texttt{CRAN}.

\section{Supplemental Materials}
The supplemental material contains proofs of the technical results stated in the paper and the Laplace approximations to evaluate the marginal likelihoods based on the nonlocal priors.
\bibliography{ref_nonlocal_jrssb,shrinkage_refs}

\begin{thebibliography}{}

\bibitem[\protect\citeauthoryear{Barber, Drton, and Tan}{Barber
  et~al.}{2016}]{barber2016laplace}
Barber, R.~F., Drton, M., and Tan, K.~M. (2016).
\newblock Laplace approximation in high-dimensional {B}ayesian regression.
\newblock In {\em Statistical Analysis for High-Dimensional Data}, pages
  15--36. Springer.

\bibitem[\protect\citeauthoryear{Bondell and Reich}{Bondell and
  Reich}{2012}]{bondell2012consistent}
Bondell, H. and Reich, B. (2012).
\newblock {Consistent high-dimensional {B}ayesian variable selection via
  penalized credible regions}.
\newblock {\em Journal of the American Statistical Association} {\bf 107,}
  1610--1624.

\bibitem[\protect\citeauthoryear{Castillo, Schmidt-Hieber, Van~der Vaart,
  et~al\mbox{.}}{Castillo et~al.}{2015}]{castillo2015bayesian}
Castillo, I., Schmidt-Hieber, J., Van~der Vaart, A., et~al. (2015).
\newblock Bayesian linear regression with sparse priors.
\newblock {\em Annals of Statistics} {\bf 43,} 1986--2018.

\bibitem[\protect\citeauthoryear{Castillo, van~der Vaart,
  et~al\mbox{.}}{Castillo et~al.}{2012}]{castillo2012needles}
Castillo, I., van~der Vaart, A., et~al. (2012).
\newblock Needles and straw in a haystack: Posterior concentration for possibly
  sparse sequences.
\newblock {\em Annals of Statistics} {\bf 40,} 2069--2101.

\bibitem[\protect\citeauthoryear{Chen and Chen}{Chen and Chen}{2008}]{Chen2008}
Chen, J. and Chen, Z. (2008).
\newblock {Extended {B}ayesian information criteria for model selection with
  large model spaces}.
\newblock {\em Biometrika} {\bf 95,} 759--771.

\bibitem[\protect\citeauthoryear{Davis and Goadrich}{Davis and
  Goadrich}{2006}]{davis2006relationship}
Davis, J. and Goadrich, M. (2006).
\newblock The relationship between {P}recision-{R}ecall and {ROC} curves.
\newblock In {\em Proceedings of the 23rd international conference on Machine
  learning}, pages 233--240. ACM.

\bibitem[\protect\citeauthoryear{Fan and Li}{Fan and
  Li}{2001}]{fan2001variable}
Fan, J. and Li, R. (2001).
\newblock Variable selection via nonconcave penalized likelihood and its oracle
  properties.
\newblock {\em Journal of the American Statistical Association} {\bf 96,}
  1348--1360.

\bibitem[\protect\citeauthoryear{Fan and Lv}{Fan and Lv}{2008}]{fan2008sure}
Fan, J. and Lv, J. (2008).
\newblock Sure independence screening for ultrahigh dimensional feature space.
\newblock {\em Journal of the Royal Statistical Society: Series B} {\bf 70,}
  849--911.

\bibitem[\protect\citeauthoryear{George and Foster}{George and
  Foster}{2000}]{george2000calibration}
George, E. and Foster, D.~P. (2000).
\newblock Calibration and empirical {B}ayes variable selection.
\newblock {\em Biometrika} {\bf 87,} 731--747.

\bibitem[\protect\citeauthoryear{Hans, Dobra, and West}{Hans
  et~al.}{2007}]{hans2007shotgun}
Hans, C., Dobra, A., and West, M. (2007).
\newblock {Shotgun stochastic search for ``large p� regression}.
\newblock {\em Journal of the American Statistical Association} {\bf 102,}
  507--516.

\bibitem[\protect\citeauthoryear{Jiang}{Jiang}{2007}]{Jiang2007}
Jiang, W. (2007).
\newblock {Bayesian variable selection for high dimensional generalized linear
  models: Convergence rates of the fitted densities}.
\newblock {\em Annals of Statistics} {\bf 35,} 1487--1511.

\bibitem[\protect\citeauthoryear{Johnson and Rossell}{Johnson and
  Rossell}{2010}]{johnson2010use}
Johnson, V.~E. and Rossell, D. (2010).
\newblock On the use of non-local prior densities in {B}ayesian hypothesis
  tests.
\newblock {\em Journal of the Royal Statistical Society: Series B} {\bf 72,}
  143--170.

\bibitem[\protect\citeauthoryear{Johnson and Rossell}{Johnson and
  Rossell}{2012}]{Johnson2012}
Johnson, V.~E. and Rossell, D. (2012).
\newblock Bayesian model selection in high-dimensional settings.
\newblock {\em Journal of the American Statistical Association} {\bf 107,}
  649--660.

\bibitem[\protect\citeauthoryear{Kim, Kwon, and Choi}{Kim
  et~al.}{2012}]{Kim2012}
Kim, Y., Kwon, S., and Choi, H. (2012).
\newblock {Consistent model selection criteria on high dimensions}.
\newblock {\em Journal of Machine Learning Research} {\bf 13,} 1037--1057.

\bibitem[\protect\citeauthoryear{Kirkpatrick and Vecchi}{Kirkpatrick and
  Vecchi}{1983}]{kirkpatrick1983optimization}
Kirkpatrick, S. and Vecchi, M. (1983).
\newblock Optimization by simulated annealing.
\newblock {\em Science} {\bf 220,} 671--680.

\bibitem[\protect\citeauthoryear{Lan, Chen, Flowers, Yandell, Stapleton, Mata,
  Mui, Flowers, Schueler, Manly, et~al\mbox{.}}{Lan
  et~al.}{2006}]{lan2006combined}
Lan, H., Chen, M., Flowers, J.~B., Yandell, B.~S., Stapleton, D.~S., Mata,
  C.~M., Mui, E. T.-K., Flowers, M.~T., Schueler, K.~L., Manly, K.~F., et~al.
  (2006).
\newblock Combined expression trait correlations and expression quantitative
  trait locus mapping.
\newblock {\em PLoS Genet} {\bf 2,} e6.

\bibitem[\protect\citeauthoryear{Liang, Paulo, Molina, Clyde, and Berger}{Liang
  et~al.}{2008}]{liang2008mixtures}
Liang, F., Paulo, R., Molina, G., Clyde, M., and Berger, J. (2008).
\newblock {Mixtures of g priors for {B}ayesian variable selection}.
\newblock {\em Journal of the American Statistical Association} {\bf 103,}.

\bibitem[\protect\citeauthoryear{Liang, Song, and Yu}{Liang
  et~al.}{2013}]{Liang2013}
Liang, F., Song, Q., and Yu, K. (2013).
\newblock {Bayesian Subset Modeling for High-Dimensional Generalized Linear
  Models}.
\newblock {\em Journal of the American Statistical Association} {\bf 108,}
  589--606.

\bibitem[\protect\citeauthoryear{Narisetty and He}{Narisetty and
  He}{2014}]{Narisetty2014}
Narisetty, N.~N. and He, X. (2014).
\newblock {Bayesian variable selection with shrinking and diffusing priors}.
\newblock {\em Annals of Statistics} {\bf 42,} 789--817.

\bibitem[\protect\citeauthoryear{Nikooienejad, Wang, and Johnson}{Nikooienejad
  et~al.}{2016}]{nikooienejad2016bayesian}
Nikooienejad, A., Wang, W., and Johnson, V.~E. (2016).
\newblock Bayesian variable selection for binary outcomes in high dimensional
  genomic studies using non-local priors.
\newblock {\em Bioinformatics} {\bf 32,} 1338--1345.

\bibitem[\protect\citeauthoryear{Radchenko, James, et~al\mbox{.}}{Radchenko
  et~al.}{2011}]{radchenko2011improved}
Radchenko, P., James, G.~M., et~al. (2011).
\newblock Improved variable selection with forward-lasso adaptive shrinkage.
\newblock {\em Annals of Applied Statistics} {\bf 5,} 427--448.

\bibitem[\protect\citeauthoryear{Raftery, Madigan, and Hoeting}{Raftery
  et~al.}{1997}]{raftery1997bayesian}
Raftery, A.~E., Madigan, D., and Hoeting, J.~A. (1997).
\newblock Bayesian model averaging for linear regression models.
\newblock {\em Journal of the American Statistical Association} {\bf 92,}
  179--191.

\bibitem[\protect\citeauthoryear{Rockova and George}{Rockova and
  George}{2014}]{rovckova2014emvs}
Rockova, V. and George, E.~I. (2014).
\newblock {EMVS}: The {EM} approach to {B}ayesian variable selection.
\newblock {\em Journal of the American Statistical Association} {\bf 109,}
  828--846.

\bibitem[\protect\citeauthoryear{Rossell and Telesca}{Rossell and
  Telesca}{2017}]{rossell2015non}
Rossell, D. and Telesca, D. (2017+).
\newblock Non-local priors for high-dimensional estimation.
\newblock {\em Journal of the American Statistical Association} .

\bibitem[\protect\citeauthoryear{Rossell, Telesca, and Johnson}{Rossell
  et~al.}{2013}]{rossell2013high}
Rossell, D., Telesca, D., and Johnson, V.~E. (2013).
\newblock High-dimensional {B}ayesian classifiers using non-local priors.
\newblock In {\em Statistical Models for Data Analysis}, pages 305--313.
  Springer.

\bibitem[\protect\citeauthoryear{Scott and Berger}{Scott and
  Berger}{2010}]{scott2010bayes}
Scott, J. and Berger, J. (2010).
\newblock {Bayes and empirical-{B}ayes multiplicity adjustment in the
  variable-selection problem}.
\newblock {\em Annals of Statistics} {\bf 38,} 2587--2619.

\bibitem[\protect\citeauthoryear{Song and Liang}{Song and
  Liang}{2015}]{song2015high}
Song, Q. and Liang, F. (2015).
\newblock High dimensional variable selection with reciprocal
  ${L}_1$-regularization.
\newblock {\em Journal of the American Statistical Association} {\bf 110,}
  1602--1620.

\bibitem[\protect\citeauthoryear{Tibshirani}{Tibshirani}{1996}]{tibshirani1996regression}
Tibshirani, R. (1996).
\newblock Regression shrinkage and selection via the lasso.
\newblock {\em Journal of the Royal Statistical Society: Series B} {\bf 58,}
  267--288.

\bibitem[\protect\citeauthoryear{Yang, Wainwright, and Jordan}{Yang
  et~al.}{2016}]{yang2015computational}
Yang, Y., Wainwright, M.~J., and Jordan, M.~I. (2016).
\newblock On the computational complexity of high-dimensional bayesian variable
  selection.
\newblock {\em Annals of Statistics} {\bf 44,} 2497--2532.

\bibitem[\protect\citeauthoryear{Yuan and Lin}{Yuan and
  Lin}{2005}]{yuan2012efficient}
Yuan, M. and Lin, Y. (2005).
\newblock Efficient empirical {B}ayes variable selection and estimation in
  linear models.
\newblock {\em Journal of the American Statistical Association} {\bf 100,}
  1215--1225.

\bibitem[\protect\citeauthoryear{Zellner}{Zellner}{1986}]{zellner1986}
Zellner, A. (1986).
\newblock On assessing prior distributions and {B}ayesian regression analysis
  with g-prior distributions.
\newblock In {\em Bayesian inference and decision techniques: Essays in Honor
  of Bruno de {F}inetti}, pages 233--243. North Holland, Amsterdam.

\bibitem[\protect\citeauthoryear{Zhang}{Zhang}{2010}]{zhang2010nearly}
Zhang, C.-H. (2010).
\newblock Nearly unbiased variable selection under minimax concave penalty.
\newblock {\em Annals of Statistics} {\bf 38,} 894--942.

\bibitem[\protect\citeauthoryear{Zhang, Lin, Zhang, et~al\mbox{.}}{Zhang
  et~al.}{2009}]{zhang2009penalized}
Zhang, D., Lin, Y., Zhang, M., et~al. (2009).
\newblock Penalized orthogonal-components regression for large p small n data.
\newblock {\em Electronic Journal of Statistics} {\bf 3,} 781--796.

\bibitem[\protect\citeauthoryear{Zou}{Zou}{2006}]{zou2006adaptive}
Zou, H. (2006).
\newblock The adaptive lasso and its oracle properties.
\newblock {\em Journal of the American Statistical Association} {\bf 101,}
  1418--1429.

\end{thebibliography}

\pagebreak
\begin{center}
\textbf{\Large Supplemental Materials to ``Scalable Bayesian Variable Selection Using Nonlocal Prior Densities in Ultrahigh-dimensional Settings"}
\end{center}
\medskip
\setcounter{equation}{0}
\setcounter{figure}{0}
\setcounter{table}{0}
\setcounter{page}{1}
\setcounter{section}{0}

\makeatletter
\renewcommand{\theequation}{S\arabic{equation}}
\renewcommand{\thefigure}{S\arabic{figure}}
\renewcommand{\bibnumfmt}[1]{[S#1]}
\renewcommand{\citenumfont}[1]{S#1}
\section{Preliminary Results}
 \begin{lemma}
\label{LemmaQ}
For $Q_\bk$ defined in \eqref{eq:Q_k}, $\prod_{j=1}^{k}Q_{\bk,j}^L\leq Q_\bk \leq \prod_{j=1}^{k}Q_{\bk,j}^U$,

where
\be 
Q_{\bk,j}^L&=&c_1(\sigma^2)^{1/2}(n\nu_\bk^*+1/\tau_{n,p})^{-1/2}\exp\{ -\tau_{n, p}/\widetilde{\beta}_{\bk,j}^{*2} \} ,\\ 
Q_{\bk,j}^U &=&c_2(\sigma^2)^{1/2}(n\nu_{\bk*}+1/\tau_{n,p})^{-1/2}\exp\{ -\tau_{n, p}/(|\widetilde{\beta}_{\bk,j}|+\widetilde{\epsilon}_n)^2 \},
\ee
and  $\widetilde{\epsilon}_n \asymp (n\nu_{\bk*}/\tau_{n, p})^{-1/4}$, with $\widetilde{\beta}_{\bk,j}^{*}\in[\widetilde{\beta}_{\bk,j}-\widetilde{\epsilon}_n,\widetilde{\beta}_{\bk,j}+\widetilde{\epsilon}_n]\setminus(-\widetilde\epsilon_n,\widetilde\epsilon_n)^c$  for some positive constants $c_1$ and $c_2$.
\end{lemma}

\begin{proof}
Recall $\widetilde{\Sigma}_{\bk} = (X_{\bk }^TX_{\bk}+1/\tau_{n,p} \mr I_{\bk})^{-1}$. From \eqref{eq:nukstar}, all eigenvalues of $(\widetilde{\Sigma}_{\bk})^{-1}$ are bounded between $n \nu_{\bk*}+1/\tau_{n,p}$ and $n\nu_{\bk }^*+1/\tau_{n,p}$, which implies for all $x \in \mathbb{R}^{|\bk|}$, $(n \nu_{\bk*}+1/\tau_{n,p}) x^T x \le x^T (\widetilde{\Sigma}_{\bk})^{-1} x \le (n\nu_{\bk }^*+1/\tau_{n,p}) x^T x$. Let $T_{1n} = \{(n \nu_{\bk}^* + 1/\tau_{n,p})/\sigma^2\}^{1/2}$ and $T_{2n} = \{(n \nu_{\bk*} + 1/\tau_{n,p})/\sigma^2\}^{1/2}$. Substituting the above inequality in the expression for $Q_{\bk}$, we have 
\begin{align}\label{eq:Qk_basicbd}
\prod_{j = 1}^{|\bk| } g_1(\widetilde{\beta}_{\bk, j}) \le Q_{\bk} \le \prod_{j = 1}^{|\bk| } g_2(\widetilde{\beta}_{\bk, j}),
\end{align}
where 
\begin{align}\label{eq:gi}
g_i(\widetilde{\beta}_{\bk, j}) = \int_{-\infty}^{\infty} \exp \{ -T_{in}^2 (\beta_{\bk, j} - \widetilde{\beta}_{\bk, j} )^2/2 - \tau_{n, p}/\beta_{\bk, j}^2 \} d\beta_{\bk, j}, 
\end{align}
for $ i = 1, 2.$
We establish the lower bound first by showing that $g_1(\widetilde{\beta}_{\bk, j}) \ge Q_{\bk, j}^L$ for all $j = 1, \ldots, |\bk|$. Recall $\widetilde{\epsilon}_n \asymp (n \nu_{\bf k*}/\tau_{n, p})^{-1/4}$ from the statement of the Lemma. We have 
\begin{eqnarray*}
g_1(\widetilde{\beta}_{\bk, j}) & \ge & \int_{[\widetilde{\beta}_{\bk,j}-\widetilde{\epsilon}_n,\widetilde{\beta}_{\bk,j}+\widetilde{\epsilon}_n]\setminus(-\widetilde\epsilon_n,\widetilde\epsilon_n)^c} \exp\{ -T_{1n}^2 (\beta_{\bk, j}-\widetilde{\beta}_{\bk, j})^2/2- \tau_{n, p}/\beta_{\bk,j}^2 \} d\beta_{\bk,j} \\
& \geq& \exp\{- \tau_{n, p}/\widetilde\beta_{\bk,j}^{*2} \} \int_{[\widetilde{\beta}_{\bk,j}-\widetilde{\epsilon}_n,\widetilde{\beta}_{\bk,j}+\widetilde{\epsilon}_n]\setminus(-\widetilde\epsilon_n,\widetilde\epsilon_n)^c}\exp\{ -T_{1n}^2 (\beta_{\bk, j}-\widetilde{\beta}_{\bk, j})^2/2 \} d\beta_{\bk,j},
\end{eqnarray*}
for some $\widetilde\beta_{\bk,j}^*\in[\widetilde{\beta}_{\bk,j}-\widetilde{\epsilon}_n,\widetilde{\beta}_{\bk,j}+\widetilde{\epsilon}_n]\setminus(-\widetilde\epsilon_n,\widetilde\epsilon_n)^c$. Then, the integral in the last line of the above display is equivalent to 
\begin{eqnarray*}
\int_{[-\widetilde{\epsilon}_n,\widetilde{\epsilon}_n]\setminus(-\widetilde\beta_{\bk,j}-\widetilde\epsilon_n,-\widetilde\beta_{\bk,j}+\widetilde\epsilon_n)^c} e^{-T_{1n}^2 t^2/2} d t  
\geq c_1T_{1n}^{-1} \int_0^{T_{1n} \widetilde{\epsilon}_n} e^{-z^2/2} dz \geq c_2 T_{1n}^{-1},
\end{eqnarray*}
where $c_1$ and $c_2$ are some positive constants and the last inequality in the above display follows since $T_{1n} \widetilde{\epsilon}_n \geq 1 $ for large $n$. Substituting back in the previous display, $g_1(\widetilde \beta_{\bf k, j}) \ge c_1 T_{1n}^{-1} \exp\{ -\tau_{n, p}/ \widetilde \beta_{\bf k, j}^{*2}\}$ for some constant $c_1 > 0$, completing the proof of the lower bound.

We now establish the upper bound by showing that $g_2(\widetilde \beta_{\bk, j}) \le Q_{\bk, j}^U$ for all $j = 1, \ldots, |\bk|$. It is straightforward to see that $g_2$ is a symmetric function (i.e, $g_2(\widetilde{\beta}_{\bk, j}) = g_2(| \widetilde{\beta}_{\bk, j} |)$), so that it is enough to establish the bound for $\widetilde{\beta}_{\bk, j} > 0$; without loss of generality we assume that $\widetilde{\beta}_{\bk, j} > 0$. We have
\begin{eqnarray*}
&& \int_{-\infty}^{\infty} \exp \{ -T_{2n}^2 (\beta_{\bk, j} - \widetilde{\beta}_{\bk, j} )^2 /2 - \tau_{n, p}\beta_{\bk, j}^2 \} d\beta_{\bk, j} \\
&=& \int_{-\infty}^{0} \exp\{ -T_{2n}^2(\beta_{\bk,j}-\widetilde{\beta}_{\bk,j})^2/2- \tau_{n, p}/\beta_{\bk,j}^2\}d\beta_{\bk,j}\\
&+& \int_{0}^{\widetilde{\beta}_{\bk,j}+\widetilde\epsilon_n}\exp\{ -T_{2n}^2(\beta_{\bk,j}-\widetilde{\beta}_{\bk,j})^2/2- \tau_{n, p}/\beta_{\bk,j}^2\}d\beta_{\bk,j}\\
&+& \int_{\widetilde{\beta}_{\bk,j}+\widetilde\epsilon_n}^{\infty}\exp\{ -T_{2n}^2(\beta_{\bk,j}-\widetilde{\beta}_{\bk,j})^2/2- \tau_{n, p}/\beta_{\bk,j}^2\}d\beta_{\bk,j}.
\end{eqnarray*}
Define the first term of the above as $W_1$, the second as $W_2$, and the third term as $W_3$. First, we shall show that $W_1 \leq c T_{2n}^{-1} \exp\{ - T_{2n} (2\tau_{n, p})^{1/2} \}$ for  some positive constant $c$. By transforming the  variable $t=\beta_{\bk,j}-\widetilde\beta_{\bk,j}$,
\be
W_1 &=& \int_{-\infty}^{0} \exp\{ -T_{2n}^2 t^2 / 2 + T_{2n}^2 t \widetilde{\beta}_{\bk,j} -T_{2n}^2 \widetilde{\beta}_{\bk,j}^2/2- \tau_{n, p}/t^2\}d t\\
&\leq& \int_{-\infty}^{0} \exp\{ -T_{2n}^2 t^2 / 2 - \tau_{n, p}/t^2\} d t\\
&\leq& c_3 T_{2n}^{-1} \exp\{ -T_{2n} (2\tau_{n, p})^{1/2}\},
\ee
for some constant $c_3$, since $\int \exp\{- \mu / t^2 - \zeta t^2 \}dt = ( \pi/ \zeta )^{ -1/2 } \exp\{ -2 (\mu \zeta)^{ 1/2 } \} $ for $\mu>0$ and $\zeta>0$.

Second, by changing the variable $z = t- \widetilde{\epsilon}$,
\be
W_2 &=& \int_{-\widetilde\epsilon_n}^{ \widetilde{ \beta}_{\bk,j} } \exp \{ -T_{2n}^2 (z-\widetilde{\beta}_{\bk,j}+\widetilde\epsilon_n)^2/2- \tau_{n, p}/(z+\widetilde\epsilon_n)^2\} d z \\
&\leq& \exp\{ - \tau_{n, p} /( \widetilde{\beta}_{\bk,j} + \widetilde\epsilon_n )^2 \} \int_{-\infty}^{\infty} \exp \{ -T_{2n}^2 (z-\widetilde{\beta}_{\bk,j}+\widetilde\epsilon_n)^2 / 2 \} \\
&\leq& c_4T_{2n}^{-1} \exp\{ - \tau_{n, p} /( \widetilde{\beta}_{\bk,j} + \widetilde{\epsilon}_n )^2 \} ,
\ee
for some positive constant $c_4$.

Third, by changing the variable $z = t- \widetilde{ \beta}_{\bk,j}$, there exists some positive constant $c$ such that 
\be
W_3 &=& \int_{\widetilde{\epsilon}_n}^{\infty} \exp\{ -T_{2n}^2 z^2 /2 -\tau_{n, p} / ( z + \widetilde{\beta}_{\bk,j} )^2 \} d z \\
&\leq& \exp \{ -T_{2n}^2\widetilde{\epsilon}_n^2/4 \} \int_{-\infty}^{\infty} \exp \{ -T_{2n}^2 z^2 / 4 \} d z \\
&\leq& c_5T_{2n}^{-1} \exp \{ -c_6 T_{ 2n } \tau_{n, p}^{1/2} \},
\ee
for some constants $c_5$and $c_6$, since $\widetilde\epsilon_n \asymp (n\nu_{\bk*}/\tau_{n, p})^{-1/4}$.
Then, 
\begin{eqnarray*}
g_2( \widetilde\beta_{\bk,j} ) \leq c_3 T_{2n}^{-1} \exp \{ -T_{2n} (2\tau_{n, p})^{1/2}\} + c_4 T_{2n}^{-1} \exp \{ - \tau_{n, p}/(\widetilde{\beta}_{\bk,j}+\widetilde{\epsilon}_n)^2 \} \\
+c_5 T_{2n}^{-1} \exp \{ - c_6 T_{2n} \tau_{n, p}^{1/2} \}.
\end{eqnarray*}
  Since $\widetilde\epsilon_n \asymp (n\nu_{\bk*}/\tau_{n,p})^{-1/4} $, when $\widetilde{\beta}_{\bk,j} < \widetilde\epsilon_n$, $\tau_{n,p}/(\widetilde{\beta}_{\bk,j}+\widetilde{\epsilon}_n)^2< \tau_{n,p}/(4\widetilde{\epsilon}_n^2) \asymp T_{2n}\tau_{n,p}^{1/2}$, and when $\widetilde{\beta}_{\bk,j} \geq \widetilde\epsilon_n$, $\tau_{n,p}/(\widetilde{\beta}_{\bk,j}+\widetilde{\epsilon}_n)^2 \leq\tau_{n,p}/(4\widetilde{\beta}_{\bk,j}^2) < T_{2n}\tau_{n,p}^{1/2}$. In overall, the right-hand side of the above display would be dominated by the second term, which { shows} that  $g_2( \widetilde\beta_{\bk,j} ) \leq  c T_{2n}^{-1}\exp\{ -\tau_{n,p}/(\widetilde\beta_{\bk,j}+\widetilde\epsilon_n)^2  \}$ for some constant $c$. When $\widetilde\beta_{\bk,j}<0$, we can show the same result by following exactly the same steps explained above.
\end{proof} 
  
\noindent We now present some auxiliary results that are used to prove Theorems 1 and 2. We make use of the following simple union bound multiple times: for non-negative random variables $V_1, \ldots, V_m$ and $a > 0$, 
\begin{align}\label{eq:union_bd}
P(\sum_{l=1}^m V_l > a) \le \sum_{l=1}^m P(V_l > a/m) \le m \max_{1 \le l \le m} P(V_l > a/m).
\end{align}


We define some notations that are used in the subsequent proofs.  Let $\bt$ denote the true data generating model, and let $\beta_\bt^0$ denote the true regression coefficient corresponding to $\bt$.
Let ${\bf c_t} = \bt \setminus \bk$, ${\bf c_k} = \bk \setminus \bt$,  and ${\bf u} = \bk \cup \bt$. Also, we define the cardinality of a model $\bk$ as $k$ and in the same spirit, denote $c_k = | { \bf c_k } |$, $c_t = | { \bf c_t } |$, and $t = | \bt |$. 
$\{ x \}_j$ denotes the {\it j}-th element of the vector $x$, and $diag\{A\}_j$ refers to the {\it j}-th diagonal element in the square matrix $A$. We denote $\chi^2_m(\lambda)$ a non-central chi-square distribution with the degrees of freedom $m$ and  non-centrality parameter $\lambda$; a central chi-square distribution is simply denoted by $\chi^2_m$.  

An important property that is used in the subsequent proofs concerns the distribution of the marginal ridge estimator. Let $\widetilde\beta_{\bk} = (X_\bk^\T X+1/\tau_{n,p}I_\bk)^{-1}X_\bk^\T  y$ and $\widetilde\beta_{\bk,j}=\{\widetilde\beta_{\bk}\}_j$. Then, 
\begin{eqnarray}\label{eq:marg_beta}
  \widetilde\beta_{\bk,j} \sim N(\beta_{\bk,j}^*,\sigma_{\bk,j}^{2*}), 
 \end{eqnarray}
where $\beta_{\bk,j}^*=\{(X_\bk^\T X+1/\tau_{n,p}I_\bk)^{-1}X_\bk^\T X_\bt\beta_\bt^*\}_j$ and $\sigma_{\bk,j}^{2*} = \sigma^2diag\{(X_\bk^\T X_\bk + 1/\tau_{n,p}I_\bk)^{-1}\}_j$. It is also evident that $(\widetilde\beta_{\bk,j}-\beta_{\bk,j}^*)^2/\sigma_{\bk,j}^{2*}\sim \chi^2_1$.\\

\noindent A set of technical results follow that are used in the proof of the main results.  Define 
\begin{align}\label{eq:H12}
H_{1n} = \sum_{\substack{\bk : \bt \subsetneq \bk,\\ |\bk| \le q_n } } \frac{ m_{\bk}(y) \pi(\bk) }{ m_{\bt}(y) \pi(\bt) } = \sum_{\substack{\bk : \bt \subsetneq \bk,\\ |\bk| \le q_n } } \frac{\pi(\bk \mid y)}{\pi(\bt \mid y) }, \quad 
H_{2n} = \sum_{\substack{\bk : \bt \nsubseteq \bk,\\ |\bk| \le q_n } } \frac{ m_{\bk}(y) \pi(\bk) }{ m_{\bt}(y) \pi(\bt) } = \sum_{\substack{\bk : \bt \nsubseteq \bk,\\ |\bk| \le q_n } } \frac{\pi(\bk \mid y)}{\pi(\bt \mid y) }.
\end{align}
\begin{lemma}\label{Lemma:over}
Fix $\epsilon > 0$. Let $\Gamma_d = \{ \bk : |\bk| \le q_n, \bt \subsetneq \bk, \, {|\bk| - |\bt|} = d \}$ for $d =1, \ldots, q_n - |\bt|$. Suppose there exist constants $c, \delta > 0$ such that $\max_{\bk \in \Gamma_d} P\big\{\pi(\bk \mid y)/\pi(\bt \mid y) > \epsilon p^{-d}/q_n \big\} \le cp^{-d(1+\delta)}$ for $d = 1, \ldots, q_n - |\bt|$. 
Then, $H_{1n}$ converges to zero in probability as $n$ tends to $\infty$, where $H_{1n}$ is as in \eqref{eq:H12}.
\end{lemma}
\begin{proof} Clearly, $|\Gamma_d| = {p - |\bt| \choose d}$. Using \eqref{eq:union_bd}, we bound
\begin{eqnarray*}
P\Big\{\sum_{\bk:\bt\subsetneq\bk}\frac{\pi(\bk\mid y)}{\pi(\bt\mid y)}>\epsilon\Big\} &=& P\Big\{\sum_{d=1}^{q_n-|\bt|}\sum\limits_{\bk\in \Gamma_{d}}\frac{\pi(\bk\mid y)}{\pi(\bt\mid y)}>\epsilon \Big\} \\
& \le & \sum_{d = 1}^{q_n - |\bt| } P \Big\{\sum_{\bk \in \Gamma_d} \frac{\pi(\bk\mid y)}{\pi(\bt\mid y)} > \epsilon/q_n \Big\} \\
&\leq&\sum\limits_{d=1}^{q_n-|\bt|}\binom{p-|\bt|}{d} \, \smash{\displaystyle\max_{ \bk\in \Gamma_{d}}} P\Big\{\frac{\pi(\bk\mid y)}{\pi(\bt\mid y)}>\epsilon p^{-d}/q_n\Big\} 
\le \sum\limits_{d=1}^{q_n-|\bt|}c p^{-d\delta}.
\end{eqnarray*}
Finally, $\sum_{d=1}^{q_n - |\bt|} c p^{-d \delta} \le c q_n p^{-\delta} \to 0$ as $n \to \infty$.
\end{proof}


\begin{lemma}
\label{Lemma:under}
Fix $\epsilon > 0$ and let $t = |\bt|$. Define $\Gamma_{k,c_k,c_t}=\{ \bk : |\bk|\leq q_n, |\bk| = k,\: |\bk \backslash \bt| = c_k, \: \: |\bt \backslash \bk|=c_t\}$ for $k = 0, \ldots, q_n; \ c_k = 0, \ldots, k; \ c_t = 1, \ldots, t$.
Suppose 
$$
\max_{ \bk \in \Gamma_{k,c_k,c_t}} P \Big[ \frac{\pi(\bk\mid y)}{\pi(\bt\mid y)}>\epsilon n^{-3}p^{-k}n^{-c_k}t^{-t} \Big] \le cp^{-k(1+\delta)},
$$
with some postive constants $c$ and $\delta$.
Then, $H_{2n}$ converges to zero as $n$ tends to $\infty$, where $H_{2n}$ is as in \eqref{eq:H12}.
\end{lemma}
\begin{proof} 
Clearly, $|\Gamma_{k,c_k,c_t}|=\binom{p}{k}\binom{k}{c_k}\binom{t}{c_t}$. 
\begin{eqnarray*}
P\Big\{\sum\limits_{\bk:\bt\nsubseteq\bk}\frac{\pi(\bk\mid y)}{\pi(\bt\mid y)}>\epsilon\Big\} &\leq&P\Big\{\sum_{k=1}^{q_n}\sum_{c_k=0}^{k}\sum_{c_t=1}^{t} \sum_{\bk\in \Gamma_{k,c_k,c_t}}\frac{\pi(\bk\mid y)}{\pi(\bt\mid y)}>\epsilon\Big\}\\
&\leq& P\Big\{\sum_{k=1}^{q_n}\sum_{c_k=0}^{k}\sum_{c_t=1}^{t} \sum_{\bk\in \Gamma_{k,c_k,c_t}}\frac{\pi(\bk \mid y)}{\pi(\bt\mid y)}>\epsilon\Big\}\\
&\leq& \sum_{k=1}^{q_n}\sum_{c_k=0}^{k}\sum_{c_t=1}^{t} P\Big\{\sum_{\bk\in \Gamma_{k,c_k,c_t}}\frac{\pi(\bk\mid y)}{\pi(\bt\mid y)}>\epsilon n^{-3}\Big\}\\
&\leq&\sum_{k=1}^{q_n}\sum_{c_k=0}^{k}\sum_{c_t=1}^{t} p^kn^{c_k}t^t\max\limits_{\bk\in \Gamma_{k,c_k,c_t}}P\Big\{\frac{\pi(\bk\mid y)}{\pi(\bt\mid y)}>\epsilon n^{-3}p^{-k}n^{-c_k}t^{-t}\Big\}\\
&\leq& \sum_{k=1}^{q_n}\sum_{c_k=0}^{k}\sum_{c_t=1}^{t} p^kn^{c_k}t^t p^{-k(1+\delta)} \rightarrow 0,
\end{eqnarray*}
as $n \rightarrow \infty$.
\end{proof}



\begin{lemma}
\label{lemma:cher}
Suppose $W$ follows a non-central chi-square distribution with the degree of freedom $m_n$ that is a positive integer and the  non-central parameter $\lambda_n\geq0$, i.e, $W \sim \chi^2_{m_n}(\lambda_n)$. Also, consider $w_n$ and $t_n$ such that $w_n\to0$ and $t_n\to\infty$ as $n$ tends to $\infty$. Also, assume that $m_n \prec t_n$. Then,
\begin{equation}\label{eq:che}
P(W\leq \lambda_n w_n) \leq c_1 \lambda_n^{-1}\exp\{  - \lambda_n (1-w_n)^2 \},
\end{equation}
And 
\begin{equation}
P(W > \lambda_n +  t_n) \leq c_2\left(\frac{t_n}{2m_n}\right)^{m_n/2}\exp\left\{ m_n/2-t_n/2  \right\} +c_3\lambda_n^{1/2}t_n^{-1}\exp\left\{  -\frac{t_n^2}{32\lambda_n} \right\} ,\label{eq:che2}
\end{equation}
where $c_1$, $c_2$, and $c_3$ are some positive constants.
\end{lemma}

\begin{proof}
 $W$ can be expressed as $W = \sum_{i=1}^{m_n}\{Z_i+(\lambda_n/m_n)^{1/2}\}^2$, where $Z_i \stackrel{i.i.d}{\sim} N(0,1)$ for $i = 1, \ldots, m$. Then, by the fact that $P(Z>a) \leq (2\pi)^{-1/2}a^{-1}\exp\{-a^2/2\}$ for any $a>0$,  we can show that there exist some positive constants $c_1$ such that  
\be
P\left(  W \leq \lambda_n w_n \right)&=& P\big\{ \sum_{i=1}^{m_n} Z_i^2 + 2(\lambda_n/m_n)^{1/2}\sum_{i=1}^{m_n} Z_i + \lambda_n \leq \lambda_n w_n \big\}\\
&\leq& P\big\{  m_n^{-1/2}\sum_{i=1}^{m_n} Z_i  \leq - \lambda_n^{1/2} (1-w_n)/2 \big\}\\
&=& P\big\{ |Z_1|  \geq \lambda_n^{1/2} (1-w_n)/2 \big\}/2\\
&\leq& c_1\lambda_n^{-1}\exp\{ -  \lambda_n (1-w_n)^2/2\},
\ee
since $Z_1$ follows a standard normal distribution.

 Also,  by using Chernoffs's bound and the fact that $P(Z>a) \leq (2\pi)^{-1/2}a^{-1}\exp\{-a^2/2\}$ for any $a>0$, one can show that 
 \begin{eqnarray*}
&& P(W > \lambda_n +  t_n) =P\left\{\sum_{i=1}^{m_n} Z_i^2 + 2 (\lambda_n/m_n)^{1/2}\sum_{i=1}^{m_n} Z_i > t_n\right\}\\
 &\leq& P\left(\sum_{i=1}^{m_n} Z_i^2 > t_n/2\right) + P\left\{m_n^{-1/2}\sum_{i=1}^{m_n} Z_i > \lambda_n^{-1/2}t_n/4\right\}\\
 &\leq& c_2\left(\frac{t_n}{2m_n}\right)^{m_n/2} \exp\left\{m_n/2-t_n/2\right\}+c_3\lambda_n^{1/2}t_n^{-1}\exp\left\{  -\frac{t_n^2}{32\lambda_n} \right\},
 \end{eqnarray*}
 where $c_2$ and $c_3$ are some positive constants.
\end{proof}

\begin{lemma}\label{lemma:Qratio}
Consider $Q_\bk$ defined in  \eqref{eq:Q_k} for an arbitrary model $\bk$.  Fix any $\delta>0$. For any $\bk$ with $ \bt \subsetneq \bk$,
 \begin{eqnarray}\label{eq:Q1}
 P\left[ Q_\bk/Q_\bt>\exp\left\{ -|\bk\setminus \bt| \tau_{n,p}^{2/3}(n\nu_{\bk*})^{1/3} +|\bt|\tau_{n,p}^{1-\delta/8}(n\nu_{\bk*})^{\delta/8}\right\} \right]\leq p^{-|\bk\setminus\bt|(1+\delta)},
  \end{eqnarray}
 and for $\bk$ such that $ \bt \nsubseteq \bk$, 
  \begin{eqnarray}\label{eq:Q2}
 P\left[ Q_\bk/Q_\bt>\exp\left\{ \norm \beta_\bt^0\norm^2_2 n\nu_{{\bf u}*} /\{2\log(\tau_{n,p}/\log p )\} \right\} \right]\leq p^{-|\bk|(1+\delta)}.
  \end{eqnarray}
  \begin{proof}
  
By Lemma \ref{LemmaQ}, it is sufficient to show that 
\begin{eqnarray}
&& P\left[\prod_{j\in\bt}(Q_{\bk,j}^U/Q_{\bt,j}^L) > \exp\{ |\bt|\tau_{n,p}^{1-\delta/8}(n\nu_{\bk*})^{\delta/8} \} \right]+P\left[ \prod_{j \in \bk\setminus\bt}Q_{\bk,j}^U >\exp\{-|\bk\setminus\bt|\tau_{n,p}^{2/3}(n\nu_{\bk*})^{1/3}\}\right]\nonumber\\
&\leq& p^{-|\bk\setminus\bt|(1+\delta)}.\label{eq:Q11}
\end{eqnarray}

We first shall show that the first term in the left-hand side of \eqref{eq:Q11} is bounded above by $\exp\{-cn\nu_{\bk*}\}$ for some constant $c$. 
\begin{eqnarray}
&&P\left[ \prod_{j\in\bt}\frac{Q_{\bk,j}^U}{Q_{\bt,j}^L} > \exp\left\{ |\bt|\tau_{n,p}^{1-\delta/8}(n\nu_{\bk*})^{\delta/8} \right\} \right] \leq \sum_{j\in\bt}P\left[ \frac{Q_{\bk,j}^U}{Q_{\bt,j}^L} > \exp\left\{ \tau_{n,p}^{1-\delta/8}(n\nu_{\bk*})^{\delta/8} \right\} \right]\nonumber\\
&=& \sum_{j\in\bt}P\left[c'\left( \frac{n\nu_{\bk*}+1/\tau_{n,p}}{n\nu_{\bt}^*+1/\tau_{n,p}} \right)^{-1/2}\exp\left\{ -\tau_{n,p}\left( 1/(|\widetilde\beta_{\bk,j}|+\widetilde\epsilon_n)^2 - 1/\widetilde\beta_{\bk,j}^{*2} \right)  \right\} > \exp\left\{ \tau_{n,p}^{1-\delta/8}(n\nu_{\bk*})^{\delta/8} \right\}  \right]\nonumber\\
&\leq& \sum_{j\in\bt}P[|\widetilde\beta_{\bk,j} - \beta_{\bk,j}^*|> \epsilon'] + \sum_{j\in\bt}P[|\widetilde\beta_{\bt,j} - \beta_{\bt,j}^*|> \epsilon'] , \label{eq:betas}
\end{eqnarray}
for some small enough $\epsilon'>0$ and some positive constant $c'$ and  $\widetilde\beta_{\bk,j}^*\in [\widetilde\beta_{\bk,j}-\widetilde\epsilon_n, \widetilde\beta_{\bk,j}+\widetilde\epsilon_n]\setminus(-\widetilde\epsilon_n,\widetilde\epsilon_n)^c$ as defined in Lemma \ref{LemmaQ}, and $\widetilde\beta_{\bk,j}$ and $\beta_{\bk,j}^*$ defined in \eqref{eq:marg_beta}.  The last inequality in the above display asymptotically holds, since 
\[
\tau_{n,p}^{1-\delta/8}(n\nu_{\bk*})^{\delta/8} \succ \tau_{n,p}/(|\beta_{\bk,j}^*|-\epsilon'-\widetilde\epsilon_n)^2,
\] 
for any $\delta>0$. 

 Since $(\widetilde\beta_{\bk,j} - \beta_{\bk,j}^*)^2/\sigma_{\bk,j}^{*2}\sim \chi^2_1$ and $\sigma_{\bk,j}^{*2}\geq(n\nu_{\bk*}+1/\tau_{n,p})^{-1}$, by using Lemma \ref{lemma:cher}, one can show that the first term in \eqref{eq:betas} bounded above by $\exp\{-c_1\epsilon'^2n\nu_{\bk*}\}$ for some constant $c_1$. Similarly,  the second term in \eqref{eq:betas} is bounded above by $\exp\{-c_2\epsilon'^2n\}$ for  some constant $c_2$, since {\it Assumption 5} states that $X_\bt^\T X_\bt/n$ is asymptically isotropic. Therefore, \eqref{eq:betas} is asymptotically bounded by $ p^{-q_n(1+\delta)}$ by {\it Assumption 3}.

Next, we shall show that the second term in the left-hand side of \eqref{eq:Q11} is bounded above by  $\exp\{-c\tau_{n,p}^{1/3}(n\nu_{\bk*})^{2/3}\}$ for some positive constant $c$. Since when $j\in\bk\setminus\bt$ and $\bt\subsetneq\bk$, $\beta^*_{\bk,j}\asymp n^{-1}$, 
\begin{eqnarray*}
&&P\left[ \prod_{j \in \bk\setminus\bt} Q_{\bk,j}^U >\exp\{-|\bk\setminus\bt|\tau_{n,p}^{2/3}(n\nu_{\bk*})^{1/3}\}\right] \\
&\leq& \sum_{j\in\bk\setminus\bt}P\left[ c' (n\nu_{\bk,j}+1/\tau_{n,p})^{-1/2}\exp\left\{-\frac{\tau_{n,p}}{(|\widetilde\beta_{\bk,j}|+\widetilde\epsilon_n)^2}\right\}>\exp\{-\tau_{n,p}^{2/3}(n\nu_{\bk*})^{1/3}\} \right]\\
&=&\sum_{j\in\bk\setminus\bt}P\left[\widetilde\beta_{\bk,j}^2>\left\{\tau_{n,p}^{1/2}\left((n\nu_{\bk*})^{1/3}\tau_{n.p}^{2/3}- \log(n\nu_{\bk*}+1/\tau_{n,p})/2+\log c'\right)^{-1/2} -\widetilde\epsilon_n \right\}^2 \right] \\
&\leq& \sum_{j\in\bk\setminus\bt} P\left[(\widetilde\beta_{\bk,j}-\beta^*_{\bk,j})^2/\sigma_{\bk,j}^* > c''\left(\frac{\tau_{n,p}}{n\nu_{\bk*}}\right)^{1/3}(n\nu_{\bk*}+1/\tau_{n,p})/\sigma^2  \right],
\end{eqnarray*}
for some positive contant $c'$ and $c''$. Since $(\widetilde\beta_{\bk,j}-\beta^*_{\bk,j})^2/\sigma_{\bk,j}^*\sim \chi_1^2$, by Lemma \ref{lemma:cher} the last quantity in the above display can be bounded by $\exp\{ -c \tau_{n,p}^{1/3}(n\nu_{\bk*})^{2/3} \}$ for some contant $c$. By {\it Assumption 3}, $\exp\{ -c \tau_{n,p}^{1/3}(n\nu_{\bk*})^{2/3} \}\prec p^{-q_n(1+\delta)}\leq p^{|\bk\setminus\bt|(1+\delta)|}$, which proves the statement \eqref{eq:Q11}. 

We now shall show that the equation \eqref{eq:Q2} holds for any $\delta>0$. The left-hand side of \eqref{eq:Q2} can be bounded above by  
\begin{eqnarray}
&&P\left[   \prod_{j \in {\bf k}}Q_{\bk,j}^U\Big(\prod_{j \in {\bt}} Q_{\bt,j}^L \Big)^{-1}  > \exp\left\{   \norm\beta_\bt^0\norm^2_2n\nu_{\bf u*}/\{2\log(\tau_{n,p}/\log p )\} \right\}\right]\nonumber\\
&\leq& \sum_{j\in\bk}P\left[ c(n\nu_{\bk*}+1/\tau_{n,p})^{-1/2}\exp\left\{ -\tau_{n,p}/(|\widetilde\beta_{\bk,j}|+\widetilde\epsilon_n)^2 \right\} > \exp\left\{ \norm\beta_\bt^0\norm^2_2n\nu_{\bf u *}/\{4|\bk|\log(\tau_{n,p}/\log p )\}  \right\} \right]\nonumber\\
&&+ \sum_{j\in\bt}P\left[ c'(n\nu_{\bt*}+1/\tau_{n,p})^{1/2}\exp\left\{ \tau_{n,p}/(\widetilde\beta_{\bt,j}^{*2}) \right\} > \exp\left\{ \norm\beta_\bt^0\norm^2_2n\nu_{\bf u *}/\{4|\bt|\log(\tau_{n,p}/\log p )\}  \right\} \right]\nonumber\\
&\leq& \sum_{j\in\bk} P\left[ -\frac{\tau_{n,p}}{(|\widetilde\beta_{\bk,j}|+\widetilde\epsilon_n)^2} >  \norm\beta_\bt^0\norm^2_2n\nu_{\bf u *}/\{4|\bk|\log(\tau_{n,p}/\log p )\} +\log c\right]\label{eq:QQ1}\\
&& + \sum_{j \in \bt} P\left[ |\widetilde\beta_{\bt,j}^*| < c'' \norm\beta_\bt^0\norm^{-1}_2(n\nu_{\bf u *})^{-1/2}\{4|\bt|\log(\tau_{n,p}/\log p )\}^{1/2}\tau_{n,p}^{1/2}  \right],\label{eq:QQ2}
\end{eqnarray}
where $c$, $c'$, and $c''$ are some positive constants.

\eqref{eq:QQ1} is always zero since the left-hand side in the probability is always negative and the right-hand side in the probability operator is always positive. So, we focus on \eqref{eq:QQ2} as below:

Since $\widetilde\beta_{\bt,j} -\widetilde\epsilon_n\leq \widetilde\beta_{\bt,j}^*\leq \widetilde\beta_{\bt,j} + \widetilde\epsilon_n$ implies $|\widetilde\beta_{\bt,j}| -\widetilde\epsilon_n \leq |\widetilde\beta_{\bt,j}^*|\leq |\widetilde\beta_{\bt,j}| + \widetilde\epsilon_n$, \eqref{eq:QQ2} can be bounded above by 
\begin{eqnarray*}
&&\sum_{j \in \bt} P\left[ |\widetilde\beta_{\bt,j}^*| < c'' \norm\beta_\bt^0\norm^{-1}_2(n\nu_{\bf u *})^{-1/2}\{4|\bt|\log(\tau_{n,p}/\log p )\}^{1/2}\tau_{n,p}^{1/2}  \right] \\
&\leq& \sum_{j \in \bt} P\left[ |\widetilde\beta_{\bt,j}|< c'' \norm\beta_\bt^0\norm^{-1}_2(n\nu_{\bf u *})^{-1/2}\{4|\bt|\log(\tau_{n,p}/\log p )\}^{1/2}\tau_{n,p}^{1/2}+\widetilde\epsilon_n  \right],
\end{eqnarray*}
where $\beta^*_{\bt,j}$ is defined in \eqref{eq:marg_beta}. Since $\widetilde\beta_{\bt,j}^2/\sigma_{\bt,j}^2\sim \chi^2_1(\beta_{\bt,j}^{*2}/\sigma_{\bt,j}^2)$ and $\sigma_{\bt,j}^2\asymp \sigma^2/n$ for $j\in\bt$, by using Lemma \ref{lemma:cher} and {\it Assumption 5}, one can show that the probability is bounded by $\exp\{-c n\}$ for some constant $c$, and it is evident that $\exp\{-c n\}\prec p^{-|\bk|(1+\delta)}$, which completes the proof of the Lemma.
  \end{proof}
\end{lemma}

\section{Proofs of Main Results}
\begin{proof}[{\bf Proof of Theorem~\ref{theo1}.}]
We have $\pi(\bf t \mid \bf y) = m_{\bt}(y) \pi(\bt)/\{ \sum_{\bk : |\bk| \le q_n} m_{\bk}(y) \pi(\bk) \}$, since $\pi(\bk) = 0$ for any $\bk$ with $|\bk| > q_n$. Recall $H_{1n}$ and $H_{2n}$ from \eqref{eq:H12} and note that $\pi(\bt \mid y) = (1 + H_{1n} + H_{2n})^{-1}$. Hence to show that $\pi(\bt \mid y)$ converges to one in probability, it is sufficient to establish that $H_{1n}$  and $H_{2n}$ both converge in probability to zero as $n$ tends to $\infty$. We shall prove the Theorem by showing: 

\noindent For any $\delta\in(0,8/3)$ and any model $\bk\in\Gamma_d$ (defined in Lemma  \ref{Lemma:over}),
\begin{eqnarray}
P\left[\frac{\pi(\bk\mid y)}{\pi(\bt\mid y)}>\epsilon p^{-d}q_n^{-1}\right] \leq p^{-d(1+\delta)},\label{eq:Hn1}
\end{eqnarray}
and for any model  $\bk\in\Gamma_{k,c_k,c_t}$ (defined in Lemma \ref{Lemma:under}),
\begin{eqnarray}
P \left[ \frac{\pi(\bk\mid y)}{\pi(\bt\mid y)}>\epsilon n^{-3}p^{-k}n^{-c_k}t^{-t} \right] \le cp^{-k(1+\delta)}.\label{eq:Hn2}
\end{eqnarray}
  Then, it is evident that $H_{1n}$ and $H_{2n}$ both converge to zero in probability by Lemma \ref{Lemma:over} and \ref{Lemma:under} respectively.
 
 First, we shall show that \eqref{eq:Hn1} holds. For any $\bk\in \Gamma_d$, recall that
\begin{eqnarray*}
&{}&P\Big[\frac{\pi(\bk\mid y)}{\pi(\bt\mid y)}>\epsilon p^{-d}q_n^{-1}\Big] \leq P\Big[ C_{n,p}^{-d}\frac{Q_\bk}{Q_\bt}\exp\Big\{-\frac{1}{2\sigma^2}\big(\widetilde{R}_\bk-\widetilde{R}_\bt\big) \Big\} >\epsilon p^{-d}/q_n \Big].
\end{eqnarray*}

Since $\widetilde{R}_\bk>R_\bk^*$ and $\widetilde{R}_\bt<R_\bt^*+\eta$,
where $\eta = d_1\widehat{\beta}_\bt^T\widehat{\beta}_\bt/\tau_{n,p}$ for some constant $d_1$ and $\widehat{\beta}_\bt$ is the ordinary least square estimator of $\beta_\bt$ in the true model \bt, by using \eqref{eq:union_bd}, the term in the last display can be bounded above by 
\begin{eqnarray}
&& P\Big[ C_{n,p}^{-d}\frac{Q_\bk}{Q_\bt}\exp\big\{-\big({R}^*_\bk-{R}^*_\bt\big)/(2\sigma^2) +\eta/(2\sigma^2)\big\} >\epsilon p^{-d}/q_n \Big]\nonumber \\
&\leq& 
P\Big[ C_{n,p}^{-d}\frac{Q_\bk}{Q_\bt}p^{d(1+\delta)+\delta} >\epsilon p^{-d}/q_n \Big]\label{eq:Qq1}\\ 
&{}&+ P \left[ R_\bt^*-R_\bk^*>2\sigma^2 d (1+\delta) \log p \right]\label{eq:Qq2} \\
&{}&+ P\left[ \exp\{\eta/(2\sigma^2)\} >\epsilon p^{\delta}\right]\label{eq:Qq3}.
\end{eqnarray}
By using Lemma \ref{lemma:Qratio}, \eqref{eq:Qq1} is less than $p^{-d(1+\delta)}$ when $\delta<8/3$. Since $(R_\bt^*-R_\bk^*)/\sigma^2 \sim \chi^2_{|\bk\setminus\bt|}$, by using  \eqref{eq:che} in Lemma \ref{lemma:cher}, 
 we can show that \eqref{eq:Qq2} is bounded by $cp^{-d(1+\delta)}$ for some positive constant $c$.
Since $\tau_{n,p} n\nu_{\bf t *}\eta/d_1 \sigma^2 \leq \widehat{\beta}_\bt^TX_\bt^TX_\bt\widehat{\beta}_\bt /\sigma^2\sim \chi_{|\bt|}^2 \left( {\beta}_\bt^{0T}X_\bt^TX_\bt{\beta}_\bt^0 \right)$, by using the inequality \eqref{eq:che2} in  Lemma \ref{lemma:cher}, \eqref{eq:Qq3} can be expressed as 
\begin{eqnarray}\nonumber
P\left[ \exp\left\{\eta/2\sigma^2 \right\} >\epsilon p^{\delta}\right] &\leq& P\left[ \tau_{n,p} n\nu_{\bf t*}\eta /d_1\sigma^2 > 2\tau_{n,p} n\nu_{\bf t *} (\log\epsilon + \delta\log p )/d_1 \right]\\ \nonumber
&\leq&P\big[ \widehat{\beta}_\bt^TX_\bt^TX_\bt\widehat{\beta}_\bt/\sigma^2> 2\tau_{n,p} n\nu_{\bf t *} (\log\epsilon + \delta\log p )/d_1 \big]\\ 
&\leq& (n\delta\log p )^{|\bt|/2}\exp\{-c_1\delta(n\log p)\} + n^{-1/2}(\delta\log p)^{-1}\exp\{ -c_2(n\log p)^2/n \}\nonumber \\
&\leq&c_3 p^{-|\bk|(1+\delta)}, \label{eq:eta} 
\end{eqnarray} 
for some positive constant $c_1$, $c_2$, and $c_3$, which proves that \eqref{eq:Hn1} holds.

Next, we consider \eqref{eq:Hn2}. Recall that $\bf u = \bk \cup \bt$. By using (\ref{eq:union_bd}), it can be shown that
\begin{eqnarray}
&&P\Big[\frac{\pi(\bk\mid y)}{\pi(\bt\mid y)}>\epsilon n^{-3}p^{-|\bk|}n^{-|\bk\setminus\bt|}|\bt|^{-|\bt|}\Big]\nonumber\\
&\leq& P\Big[C_{n,p}^{-(|\bk|-|\bt|)} \frac{Q_\bk}{Q_\bt}\exp\left\{-(\widetilde{R}_\bk-\widetilde{R}_\bt)/(2\sigma^2) \right\} > \epsilon n^{-3}p^{-|\bk|}n^{-|\bk\setminus\bt|}|\bt|^{-|\bt|} \Big]\nonumber\\
&\leq& P\Big[ C_{n,p}^{-|\bk|-|\bt|)}\frac{Q_\bk}{Q_\bt} \exp\left\{ -(R_\bk^* -  R_{\bf u}^*  )/(2\sigma^2) \right\} > n^{-3-|\bk\setminus\bt|}|\bt|^{-|\bt|}p^{ - |\bk|(2+\delta) +\delta} \Big]\nonumber\\
&{}& +P\big[ \exp\big\{\big( R_\bt^* - R_{\bf u}^* \big)/(2\sigma^2) \big\} \geq \epsilon p^{ |\bk|(1+\delta) } \big] + P\big[ \exp\big( \eta/(2\sigma^2) \big)>p^{ \delta } \big]\nonumber\\
&\leq&P\big[ \exp\big\{\big( R_\bt^* - R_{\bf u}^* \big)/2\sigma^2 \big\} > \epsilon p^{ |\bk|(1+\delta) } \big]\label{eq:QQQ1}\\
&{}& + P\big[ \exp\big( \eta/2\sigma^2 \big)>p^{ \delta } \big]\label{eq:QQQ2}\\
&{}& + P\left[ R^*_\bk-R_{\bf u}^* < 2\sigma^2\norm\beta_\bt^0\norm^2_2n\nu_{\bf u*}/\log(\tau_{n,p}/\log p )\right]\label{eq:QQQ3}\\
&{}& + P\left[   Q_\bk/Q_\bt  > \exp\left\{   \norm\beta_\bt^0\norm^2_2n\nu_{\bf u*}/\{2\log(\tau_{n,p}/\log p )\} \right\}\right].\label{eq:QQQ4}
\end{eqnarray}
 Since $( R_\bt^*-R_{\bf u}^*)/\sigma^2$ follows a $\chi^2_{|{\bf u}\setminus\bt|}$ distribution, \eqref{eq:QQQ1} is also bounded by $c_1p^{-|\bk|(1+\delta)}$ with some constant $c_1$. By following the same steps regarding \eqref{eq:eta}, one can show that \eqref{eq:QQQ2} is bounded by $c_2p^{-|\bk|(1+\delta)}$ for some constant $c_2$.
We note that  $(R_\bk^* - R_{\bf u}^*) / \sigma^2 \sim \chi_{|{\bf u}\setminus\bk|}^2(\lambda_n)$ with $\lambda_n=\beta_{\bt}^{0T} X_\bt^T(P_{\bf u}-{P}_\bk)X_\bt \beta^0_{\bt}$,  where  ${ P}_\bk$ is the projection matrix of $X_\bk$.  As discussed in \citet{Narisetty2014}, $\lambda_n\geq n\nu_{\bf u* } \norm \beta^0_\bt \norm^2_2$. Hence, by using {Lemma \ref{lemma:cher}}, one can show that \eqref{eq:QQQ3} is bounded by $\exp\{ - c_3\norm \beta_\bt^0 \norm^2_2n\nu_{\bf u*}/\log(\tau_{n,p}/\log p ) \}$ for some constant $c_3$.  Lemma \ref{lemma:Qratio} states that \eqref{eq:QQQ4} is bounded by $p^{-|\bk|(1+\delta)}$. 
In summary, since $q_n \prec \tau_{n,p}/\log p$ by {\it Assumption 3}, there exists some positive  constant $c_4$ such that $P[\pi(\bk\mid y)/\pi(\bt\mid y)>\epsilon n^{-3}p^{-|\bk|}n^{-|\bk\setminus\bt|}|\bt|^{-|\bt|}]\leq c_4 p^{-|\bk|(1+\delta)}$.  
which completes the proof of Theorem \ref{theo1}.
\end{proof}

\begin{proof}[{\bf Proof of Corollary~\ref{cor1}}.]

 Recall the penalty term of a model {\bk},  $Q_\bk^*$, based on the piMoM priors is 
\[Q_\bk^* = \int \exp\big\{-(\beta_\bk-\widehat{\beta}_\bk)^T{\Sigma}_\bk^{*-1}(\beta_\bk-\widehat{\beta}_\bk)/(2\sigma^2) - \sum_{j=1}^{|\bk|} \tau_{n, p}/\beta_{\bk,j}^2  - r \sum_{j=1}^{|\bk|}\log(\beta_{\bk,j}^{2}) \big\} d\beta_\bk,
\] 
in \eqref{eq:Q_k2}. Since, for any $\epsilon>0$, $\exp\big[ -\sum_{j=1}^{|\bk|}\{ \epsilon \tau_{n,p}/\beta_{\bk,j}^2+r\log (\beta_{\bk,j}^2)\} \big]$ is bounded above with respect to $\beta_{\bk,j}$, $Q_\bk^* \leq C \int  \exp\{-(\beta_\bk-\widehat{\beta}_\bk)^T{\Sigma}_\bk^{*-1}(\beta_\bk-\widehat{\beta}_\bk)/(2\sigma^2) - \sum_{j=1}^{|\bk|} (1-\epsilon)\tau_{n, p}/\beta_{\bk,j}^2 \} d\beta_\bk$ for some constant $C$. Following the exactly same steps in {Lemma \ref{LemmaQ}}, $Q_\bk^* \leq C'(n\nu_\bk^*)^{-1/2}\prod_{j=1}^{|\bk|} \exp\{ -(1-\epsilon)\tau_{n, p}/(|\widehat{\beta}_{\bk,j}| + \widetilde\epsilon_n)^2 \} $ for some constant $C'>0$.

We shall show that  the model selection procedure based on piMoM priors as in \eqref{eq:mod_piMoM} assures consistency by proving that $Q_\bk^*$ and $Q_\bk$ are asymptotically equivalent. 
 
 Next, we shall show that $Q_\bk^*$ is bounded below by   $C(n\nu_\bk^*)^{-1/2}\prod_{j=1}^{|\bk|} \exp\{ -(1-\epsilon)\tau_{n, p}/\widehat{\beta}_{\bk,j}^{*2} \} $ for some constant $C>0$ and $\widehat{\beta}_{\bk,j}^{*} \in [ \widehat{\beta}_{\bk,j} -\widetilde\epsilon_n , \widehat{\beta}_{\bk,j} + \widetilde\epsilon_n ]$. Since $\exp\big\{ - \epsilon \tau_{n,p}/\beta_{\bk,j}^2+r\log (\beta_{\bk,j}^2)\big\}$ can be minimized in $ [ \widehat{\beta}_{\bk,j} -\widetilde\epsilon_n , \widehat{\beta}_{\bk,j} + \widetilde\epsilon_n ]$, by following the proof of {Lemma \ref{LemmaQ}},
 \begin{eqnarray*}
 &&\int_{-\infty}^{\infty} \exp\{ -n\nu_{\bk}^* (\beta - \widehat\beta_{\bk,j})^2/(2\sigma^2) -\tau_{n,p}/\beta^2 - r\log( \beta^2 ) \} d \beta \\
 &\geq& \int_{\widehat\beta_{\bk,j}-\widetilde\epsilon_n}^{\widehat\beta_{\bk,j}+\widetilde\epsilon_n}\exp\{ -n\nu_{\bk}^* (\beta - \widehat\beta_{\bk,j})^2/(2\sigma^2)  -  (1-\epsilon) \tau_{n,p}/\beta^2\}  \exp\{  - \epsilon \tau_{n,p}/\beta^2 - r\log( \beta^2 ) \} d \beta \\
 &\geq& C(n\nu_\bk^*)^{-1/2} \exp\left\{ -(1-\epsilon)\tau_{n, p}/\widehat{\beta}_{\bk,j}^{*2}  \right\},
\end{eqnarray*}
  where $C$ is some constant and $\widehat{\beta}_{\bk,j}^{*} \in [ \widehat{\beta}_{\bk,j} -\widetilde\epsilon_n , \widehat{\beta}_{\bk,j} + \widetilde\epsilon_n ]\setminus(-\widetilde\epsilon_n,\widetilde\epsilon_n)^c$. 
  
  Therefore, due to the asymptotic similarity between the ridge estimator and the least square estimator, the lower and upper bounds of $Q_\bk^*$ are asymptotically equivalent to those of $Q_\bk$ with the penalty parameter $(1-\epsilon)\tau_{n,p}$, which assures the strong consistency of the model selection based on the piMoM priors. 
\end{proof}

\begin{proof}[{\bf Proof of Theorem~\ref{theo2}.}]
Under a situation where $\sigma^2$ is unknown, it is clear that 
\begin{eqnarray*}
m_\bk(y) = \tau_{n,p}^{-\frac{|\bk|}{2}} \int (2\pi\sigma^2)^{-\frac{n+|\bk|}{2}} \int\exp\left\{ |\bk|\left(\frac{2}{\sigma^2}\right)^{1/2}  -\frac{(\beta_\bk-\widetilde\beta_\bk)^T \widetilde\Sigma_\bk^{-1}(\beta_\bk-\widetilde\beta_\bk)}{2\sigma^2} - \sum_{j=1}^{|\bk|} \frac{\tau_{n,p}}{\beta_{\bk,j}^2} \right\}\pi(\sigma^2)d \beta_\bk d\sigma^2,
\end{eqnarray*}
where $\pi(\sigma^2)$ is the prior for $\sigma^2$ (Inverse-gamma density with hyperparameters $a_0$ and $b_0$). 

First, we shall show that the ratio between marginal likelihoods of a model $\bk$ and the true model $\bt$ can be bounded as
\begin{eqnarray}
\frac{m_\bk(y)}{m_\bt(y)} \leq c^{\frac{|\bk|-|\bt|}{2}}\left(  \frac{ \widetilde R_\bk +2b_0  }{\widetilde R_\bt +2b_0} \right)^{-n/2-a_0}\exp\left\{ - \sum_{j=1}^{|\bk|} \frac{\tau_{n,p}}{(|\widetilde\beta_{\bk,j}| + \widetilde\epsilon_n)^2} +\sum_{j=1}^{|\bt|} \frac{\tau_{n,p}}{\widetilde\beta_{\bt,j}^{*2}} \right\}\frac{(n\nu_{\bk*}\tau_{n,p}+1)^{-|\bk|/2}}{(n\nu_{\bt}^*\tau_{n,p}+1)^{-|\bt|/2}},\nonumber\\
\label{eq:sigma}
\end{eqnarray}
where $\widetilde{\beta}_{\bt,j}^{*}\in[\widetilde{\beta}_{\bt,j}-\widetilde{\epsilon}_n,\widetilde{\beta}_{\bt,j}+\widetilde{\epsilon}_n]\setminus(-\widetilde\epsilon_n,\widetilde\epsilon_n)^c$ for $j\in 1,\dots,|\bt|$ and $c$ is some constant. Next, we shall show that $\{ (\widetilde R_\bk +2b_0 ) /(\widetilde R_\bt +2b_0) \}^{-n/2-a_0}\leq\exp\{ - (\widetilde R_\bk - \widetilde R_{\bt})/(2\sigma^2_0(1+u_n)) \}$, where $\sigma^2_0$ is the true regression variance that involves in the data-generating process, and $u_n$ is some random variable that is concentrated around a finite value with at least probability $1-\exp\{-cn\}$ for some constant $c$. Then, by following the same steps in the proof of Theorem \ref{theo1}, the proof of Corollary \ref{cor1} is completed.  

By Lemma \ref{LemmaQ}, the marginal likelihood of a model $\bk$ can be bounded by
\begin{eqnarray*}
m_\bk(y) &\leq& \{c_1(n\nu_{\bk*}\tau_{n,p}+1)\}^{-\frac{|\bk|}{2}}\int (\sigma^2)^{-\frac{n+2a_0}{2}-1}\exp\left\{ |\bk|\left( \frac{2}{\sigma^2} \right)^{1/2} - \sum_{j=1}^{|\bk|} \frac{\tau_{n,p}}{(|\widetilde\beta_{\bk,j}| + \widetilde\epsilon_n)^2}-\frac{\widetilde R_\bk + 2b_0}{2\sigma^2} \right\} d\sigma^2\\
&\leq&\{c_1(n\nu_{\bk*}\tau_{n,p}+1)\}^{-\frac{|\bk|}{2}}\exp\left\{ - \sum_{j=1}^{|\bk|} \frac{\tau_{n,p}}{(|\widetilde\beta_{\bk,j}| + \widetilde\epsilon_n)^2}\right\} (1+\exp\{ 2|\bk| \})\left( \widetilde R_\bk + 2b_0 \right)^{-\frac{n+2a_0}{2}},
\end{eqnarray*} 
for some constant $c_1$.

Also, by using Lemma \ref{LemmaQ}, one can show that
\begin{eqnarray*}
m_\bk(y) &\geq& \{c_2(n\nu_{\bk*}\tau_{n,p}+1)\}^{-\frac{|\bk|}{2}}\int (\sigma^2)^{-\frac{n+2a_0}{2}-1}\exp\left\{ |\bk|\left( \frac{2}{\sigma^2} \right)^{1/2} - \sum_{j=1}^{|\bk|} \frac{\tau_{n,p}}{\widetilde\beta_{\bk,j}^{*2}}-\frac{\widetilde R_\bk + 2b_0}{2\sigma^2} \right\} d\sigma^2\\
&\geq&\{c_2(n\nu_{\bk*}\tau_{n,p}+1)\}^{-\frac{|\bk|}{2}}\exp\left\{ - \sum_{j=1}^{|\bk|} \frac{\tau_{n,p}}{\widetilde\beta_{\bk,j}^{*2} }\right\} \left( \widetilde R_\bk + 2b_0 \right)^{-\frac{n+2a_0}{2}},
\end{eqnarray*} 
where $c_2$ is some constant and $\widetilde{\beta}_{\bk,j}^{*}\in[\widetilde{\beta}_{\bk,j}-\widetilde{\epsilon}_n,\widetilde{\beta}_{\bk,j}+\widetilde{\epsilon}_n]\setminus(-\widetilde\epsilon_n,\widetilde\epsilon_n)^c$ for $j\in 1,\dots,|\bk|$. These results shows that \eqref{eq:sigma} holds.

Next, we consider the asymptotic behavior of $\{ (\widetilde R_\bk +2b_0 ) /(\widetilde R_\bt +2b_0) \}^{-n/2-a_0}$ in \eqref{eq:sigma}.
Define $\rho_n$ as the follows:
\[
\rho_n = (\widetilde R_\bt+2b_0)/(n\sigma_0^2)-1.
\] 
Since $-\log(1-u)<u/(1-u)$ for $u\in\mathbb{R}$,
\begin{eqnarray*}
- \log\{(\widetilde R_\bk+2b_0)/(\widetilde R_\bt+2b_0)\}&=&-\log[ 1+(\widetilde R_\bk-\widetilde R_\bt)/\{n(1+\rho_n)\sigma_0^2\} ]\\
&\leq&(\widetilde R_\bt-\widetilde R_\bk)/\{n\sigma_0^2(1+u_n)\}, 
\end{eqnarray*}
where $u_n = \rho_n+(\widetilde R_\bk-\widetilde R_\bt)/(n\sigma_0^2)$. 

 Since $(R^*_\bk- R^*_{\bf u})/\sigma_0^2\sim \chi_{|{\bf u}\setminus\bk|}(\lambda_n)$ with  $\lambda_n=\beta_{\bt}^{0T}X_{\bt}^T(P_{\bf u}-P_\bk)X_{\bt}\beta_{\bt}^0/\sigma^2_0$,  by using Lemma \ref{lemma:cher} one can show that 
\begin{eqnarray*}
P\left( \left| u_n- \lambda_n/n \right|>\epsilon\right)&\leq& P\left( |\rho_n|>\epsilon/4 \right) +P\left\{ (R_\bt^*-R_{\bf u}^*)/(n\sigma_0^2)>\epsilon/4\right\} \\&&+ P\left\{ \left| (R^*_\bk- R^*_{\bf u})/(n\sigma_0^2)-\lambda_n/n\right|>\epsilon/4\right\}+P\left( \eta/2n\sigma_0^2 > \epsilon/4 \right)\\
&\leq& \exp\{-c'n\}+P\left\{ \left|(R^*_\bk- R^*_{\bf u})/(n\sigma_0^2)-\lambda_n/n\right|>\epsilon/4\right\}\\
&\leq&\exp\{-c^{\prime\prime}n\},
\end{eqnarray*}
for some constant $c'$ and $c^{\prime\prime}$, and $\eta$ is defined in the proof of { Theorem \ref{theo1}}. Also, by Assumption 5, $\lambda_n/n$ will be bounded below and above.
\end{proof}

\begin{proof}[{\bf Proof of Corollary~\ref{cor2}.}]
 Since we showed that the asymptotic equivalence between $Q_\bk$ and $Q_\bk^*$ in the proof of Corollary \ref{cor1}, by following exactly same steps in the proof of Theorem \ref{theo2} we can prove the model selection consistency under piMoM prior densities.
\end{proof}

\begin{proof}[{\bf Proof of Proposition~\ref{lemma:rLASSO}}.]

  We shall show that for any $\alpha_\bk = \widehat\beta_\bk+\bige_n$ with $\bige_n=\{\epsilon_{n,j}\}_{j=1,\dots,|\bk|}$ and $|\epsilon_{n,j}| \succ \epsilon_n^*$ for at least one $j\in\{1,\dots,|\bk|\}$, $P\{g(\alpha_\bk ; \bk) < g(\widetilde\beta_\bk^* ; \bk)\}\to 0$ as $n$ tends to $\infty$, where $\widetilde\beta_\bk^* \in {B}(\widehat\beta_\bk ; \epsilon_n^*)$ with $\epsilon_n^*\asymp (\tau_{n,p}/n)^{1/3}$. More specifically, we set $\widetilde\beta_{\bk,j}^* = \widehat\beta_{\bk,j}+\epsilon_n^*$ for $j\in \bt$ and $\widetilde\beta_{\bk,j}^* = \widehat\beta_{\bk,j}$ for $j\in \bt^c$. 
  Without loss of generality, we assume that $X_j^\T X_j = n$ for $j=1,\dots,p$.
  
    Note that
  \begin{eqnarray*}\nonumber 
  g(\alpha_\bk ; \bk) &=&  || X_\bk\alpha_\bk -  X_\bk\widehat\beta_\bk||^2_2 + \sum_{j=1}^{|\bk|} \tau_{n,p}/|\alpha_{\bk,j}| + D_n\\ \nonumber 
  &=&    \sum_{j=1}^{|\bk|} \{ c_jn \epsilon_{n,j}^2  + \tau_{n,p}/|\widehat\beta_{\bk,j}+\epsilon_{n,j}|  \}+ D_n, \\ 
  \end{eqnarray*}
  for some constants $c_j$ such that $C_L<c_j<C_U$ for $j=1,\dots,|\bk|$, and some randome variable $D_n$ that are not relevant to $\alpha_\bk$. 
Then,
\begin{eqnarray}
&&P\{g(\alpha_\bk ; \bk) < g(\widetilde\beta_\bk^* ; \bk)\}\nonumber\\
&\leq&P\left[  \sum_{j=1}^{|\bk|} \left\{ c_jn \epsilon_{n,j}^2  + \frac{\tau_{n,p}}{|\widehat\beta_{\bk,j}+\epsilon_{n,j}|}  \right\} <  \sum_{j=1}^{|\bk|} \left\{ c_jn \epsilon_n^{*2}  +\frac{\tau_{n,p}}{|\widetilde\beta_{\bk,j}^*|} \right\} \right]\nonumber\\ 
&\leq& P\left[\sum_{j\in S^*\cap S_{\bk,n}}\left\{ c_jn \epsilon_{n,j}^2  + \frac{\tau_{n,p}}{|\widehat\beta_{\bk,j}|+|\epsilon_{n,j}|}-t_{n,j} \right\}< \sum_{j\in S^*\cap S_{\bk,n}} \left\{c_jn\epsilon_n^{*2}  + \frac{\tau_{n,p}}{|\widetilde\beta_{\bk,j}^*|} \right\}\right]\label{eq:rlasso1}\\
&&+P\left[ \sum_{j\in S^{*}\cap S_{\bk,n}^c} \left\{ c_jn \epsilon_{n,j}^2  + \frac{\tau_{n,p}}{|\widehat\beta_{\bk,j}|+|\epsilon_{n,j}|}- t_{n,j}\right\} < \sum_{j\in S^{*}\cap S_{\bk,n}^c}\left\{ c_jn\epsilon_n^{*2}  +\frac{\tau_{n,p}}{|\widetilde\beta_{\bk,j}^*|}\right\} \right]\label{eq:rlasso2}\\
&&+P\left[\sum_{j\in S^{*c}}\left\{ c_jn \epsilon_{n,j}^2  + \frac{\tau_{n,p}}{|\widehat\beta_{\bk,j}|+|\epsilon_{n,j}|} + \sum_{j \in S^*} \frac{t_{n,j}}{|S^{*c}|}\right\} <  \sum_{j\in S^{*c}}\left\{c_jn\epsilon_n^{*2}  +\frac{\tau_{n,p}}{|\widetilde\beta_{\bk,j}^*|} \right\}\right],
\label{eq:rlasso3}
\end{eqnarray}
where $t_n$ is an arbitrary sequence such that $t_{n,j} = n^{2/3}\tau_{n,p}^{1/3}\epsilon_{n,j}$, and $S^*=\{j\in\{1,\dots,p\}:|\epsilon_{n,j}| \succ \epsilon_n^*\}$, and $S_{\bk,n} = \{ j\in \bk : |\widehat\beta_{\bk,j}| < \epsilon_n^* \}$. Then, to complete the proof, it is sufficient to show that each of \eqref{eq:rlasso1}, \eqref{eq:rlasso2}, and \eqref{eq:rlasso3} converges to zero.

Since $n(\widehat\beta_{\bk,j}-\beta_{\bt,j}^0)^2/\sigma^{2} \sim \chi^2_1$ for $j= 1,\dots,|\bk|$,
 \[
 P(|\widehat\beta_{\bt,j}-\beta_{\bt,j}^0|>\zeta_n)\leq(\pi n\zeta_n^2/2)^{-1/2}\exp\{ -n\zeta_n^2/(2\sigma^2) \},
 \] for any $\zeta_n>0$. This implies that 
 $S_{\bk,n} = \bt$ at least probability $1-|\bt^c|(\pi n\epsilon_n^{*2}/2)^{-1/2}\exp\{ -n\epsilon_n^{*2}/(2\sigma^2) \}$. Therefore,  the equation \eqref{eq:rlasso1} can be asymptotically bounded by
 \begin{eqnarray*}
 &&\sum_{j\in S^*\cap \bt}P\left[ c_jn \epsilon_{n,j}^2  + \frac{\tau_{n,p}}{2|\epsilon_{n,j}|}-t_{n,j} <  c_jn\epsilon_n^{*2}  +  \frac{\tau_{n,p}}{|\widehat\beta_{\bk,j}+\epsilon_n^*|}\right]\\
 &\leq& \sum_{j\in S^*\cap \bt}P\left[ |\widehat\beta_{\bk,j}+\epsilon_n^*|<c\tau_{n,p}(n\epsilon_{n,j}^2-t_{n,j}+\tau_{n,p}/|\epsilon_{n,j}|)^{-1} \right],
 \end{eqnarray*}
 for some positive constant $c$. Consider Lemma \ref{lemma:cher} with $\lambda_n = n\epsilon_n^{*2}/\sigma^2$ and $w_n = c^2\tau_{n,p}^2/\{\epsilon^{*2}_n(n\epsilon_{n,j}^2-t_{n,j}+\tau_{n,p}/|\epsilon_{n,j}|)^2\}$ for $j\in S^*\cap \bt$. Since $n\epsilon_{n,j}^2 \succ n^{1/3}\tau_{n,p}^{2/3}$ for $j\in S^*$ implies $w_n\to 0$, Lemma \ref{lemma:cher} guarantees that the last display is bounded by $c'|S^*\cap \bt|\lambda_n^{-1}\exp\{-\lambda_n(1-w_n)^2\}$ for some constant $c'$, which means that \eqref{eq:rlasso1} converges to zero as $n$ tends to $0$. By following the same steps, one can show that \eqref{eq:rlasso2} converges to zero.
 
 Also, \eqref{eq:rlasso3} can be asymptotically bounded by
 \begin{eqnarray*}
 &&\sum_{j\in S^{*c}\cap\bt}P\left[ c_jn \epsilon_{n,j}^2  +  \frac{\tau_{n,p}}{2|\epsilon_{n,j}|} +c\min_{j\in S^*} t_{n,j}<  c_jn\epsilon^{*2}+  \frac{\tau_{n,p}}{|\widehat\beta_{\bk,j}+\epsilon_n^*|}\right]\\
 &&+\sum_{j\in S^{*c}\cap\bt^c}P\left[ c_jn \epsilon_{n,j}^2  +  \frac{\tau_{n,p}}{2|\widehat\beta_{\bk,j}+\epsilon_n^*|} +c\min_{j\in S^*} t_{n,j} <  c_jn\epsilon^{*2}+  \frac{\tau_{n,p}}{|\widehat\beta_{\bk,j}+\epsilon_n^*|}\right]\\
 &\leq& \sum_{j\in S^{*c}\cap \bt}P\left[ |\widehat\beta_{\bk,j}+\epsilon_n^*|<c'\tau_{n,p}(n\epsilon_{n,j}^2-n\epsilon_{n}^{*2}+c\min_{j\in S^*} t_{n,j}+\tau_{n,p}/|\epsilon_{n,j}|)^{-1} \right]\\
 &&+\sum_{j\in S^{*c}\cap \bt^c}P\left[ |\widehat\beta_{\bk,j}+\epsilon_n^*|<c''\tau_{n,p}(n\epsilon_{n,j}^2-n\epsilon_{n}^{*2}+c\min_{j\in S^*} t_{n,j}+\tau_{n,p}/|\epsilon_{n,j}|)^{-1}/2 \right],
\end{eqnarray*} 
 where $c$, $c'$, and $c''$ are some positive constants. For the first term in the last line of the above display, by setting $\lambda_n = n\epsilon^{*2}/\sigma^2$ and $w_n =c^2\tau_{n,p}^2/\{\epsilon^{*2}_n(n\epsilon_{n,j}^2 - n\epsilon_n^* +c\min_{j\in S^*} t_{n,j}+\tau_{n,p}/|\epsilon_{n,j}|)^2\}$, we can apply Lemma \ref{lemma:cher}. Since  $w_n \prec \tau_{n,p}^2(\epsilon_n^*\min_{j\in S^*} t_{n,j})^{-2}$ implies  $w_n\to 0$, the first term in the above display converges to zero by Lemma \ref{lemma:cher}. Similarly, the second term also converges to zero.   
\end{proof}
 
\section{Laplace Approximations of Marginal Likelihoods}
In this section, we provide the Laplace approximation of the marginal likelihoods based on the nonlocal priors. Because closed form expressions for posterior model probabilities based on modified peMoM priors and modified piMoM priors are not available, we estimate the posterior model probabilities using Laplace approximations.  For posterior probabilities based on the peMoM priors, an inverse-Gamma density with parameters $(a_0,b_0)$ on  $\sigma^2$ the Laplace approximation to the marginal density of the data for model $\bk$ can be expressed as
\begin{eqnarray}\label{laplace}
\pi(\bk\mid y) \propto (2\pi)^{|\bk|/2} \left| V(\beta^*_\bk,\sigma^{2*}) \right| ^{-1/2}\exp\{ f(\beta^*_\bk,\sigma^{2*}) \}p(\bk),
\end{eqnarray}
where
\begin{eqnarray*}
(\beta^*_\bk,\sigma^{2*}) &=& \underset{(\beta_\bk,\sigma^2)} {\mathrm{argmax}}f(\beta_\bk,\sigma^2)\\
f(\beta_\bk,\sigma^2) &=& -\left(n/2+|\bk|/2+a_0+1 \right)\log \sigma^2 -(y-X_\bk\beta_\bk)^T(y-X_\bk\beta_\bk)/(2\sigma^2)-\beta_\bk^T\beta_\bk/(2\sigma^2\tau_{n,p})\\
&&-\sum_{j=1}^{|\bk|}\tau_{n, p}/\beta_{\bk,j}^2 + |\bk|(2/\sigma^2)^{1/2} - b_0/\sigma^2+|\bk|(\log\tau_{n,p})/2,
\end{eqnarray*}
and $ V(\beta_\bk,\sigma^{2})$ is a $(|\bk|+1)\times(|\bk|+1)$ matrix with the following blocks:
\begin{eqnarray*}
V_{11} &=& X_\bk^T X_\bk/\sigma^2+ I_\bk/\sigma^2\tau_{n,p}+diag\left\{ 6\tau_{n, p}/\beta_{\bk,j}^4 \right\}_{j=1,\ldots,|\bk|}\\
V_{12}&=& X_\bk^T(y-X_\bk\beta_\bk)/\sigma^4-\beta_\bk/\{\sigma^4\tau_{n,p}\}\\
V_{22}&=& -(n/2+|\bk|/2+a_0+1)/\sigma^4+ (y-X_\bk\beta_\bk)^T(y-X_\bk\beta_\bk)/\sigma^6-\beta_\bk^T\beta_\bk/\tau_{n,p} \\
&&-3|\bk|2^{1/2}\sigma^{-5}/4+2b_0/\sigma^6.
\end{eqnarray*}
For the piMoM priors on $\beta_\bk$,  the Laplace approximation of the posterior model probability can be expressed as in \eqref{laplace}, but with
\begin{eqnarray*}
f(\beta_\bk,\sigma^2) &=& -\left(n/2 + a_0+1 \right)\log \sigma^2 -(y-X_\bk\beta_\bk)^T(y-X_\bk\beta_\bk)/(2\sigma^2)\\
&&-\sum_{j=1}^{|\bk|}\big\{ r\log(\beta_{\bk,j}^2)+ \tau_{n, p}/\beta_{\bk,j}^2  \big\} + |\bk|\big\{ (r-1/2)\log \tau_{n,p} -\log \Gamma(r-1/2) \big\} - b_0/\sigma^2,
\end{eqnarray*}
and $ V(\beta_\bk,\sigma^{2})$ a $(|\bk|+1)\times(|\bk|+1)$ matrix with the following blocks:
\begin{eqnarray*}
V_{11} &=&X_\bk^TX_\bk/\sigma^2 +diag\left\{ 6\tau_{n, p}/\beta_{\bk,j}^4 - 2r/\beta_{\bk,j}^2 \right\}_{j=1,\ldots,|\bk|}\\
V_{12}&=& X_\bk^T(y-X_\bk\beta_\bk)/\sigma^4\\
V_{22}&=& -(n/2+a_0+1)/\sigma^4+ (y-X_\bk\beta_\bk)^T(y-X_\bk\beta_\bk)/\sigma^6
+2b_0/\sigma^6.
\end{eqnarray*}

\end{document}